\documentclass[12pt,english,twoside]{article}

\usepackage[a4paper]{geometry} 
\usepackage[T1]{fontenc}
\usepackage{lmodern}
\usepackage{microtype}
\usepackage[english]{babel} 
\usepackage{datenumber,fp}  
\usepackage{paralist,hyperref,fancyhdr}
\usepackage[indentafter]{titlesec}
\usepackage{blindtext,lipsum} 
\usepackage[toc,page]{appendix} 
\usepackage{booktabs} 
\usepackage{etoolbox} 
\usepackage[skip=10pt]{caption} 
\usepackage{lineno} 
\usepackage{graphicx,subfigure,import} 
\usepackage{tikz,pgfplots}   
\usepackage{float} 
\usepackage{scalerel}
\usepackage[ansinew]{inputenc}
\usepackage{mathtools,amsmath,amsfonts,mathrsfs,amssymb,amsthm,bbm,bm}
\usepackage{indentfirst,csquotes,dsfont,xcolor}
\usepackage{enumitem}
\usepackage[square,numbers]{natbib}
\usepackage{algpseudocode} 
\usepackage{algorithm}

\setlist[enumerate]{leftmargin=2.25em, labelsep=*}
\pgfplotsset{compat=1.18} 
\emergencystretch=\maxdimen
\hypersetup{ colorlinks=true, linkcolor=black, filecolor=black, urlcolor=black }

\allowdisplaybreaks

\titleformat{name=\section}{}{\thetitle.}{0.8em}{\centering\scshape}
\titlespacing*{\section}{2pt}{6ex}{2ex}
\titleformat{name=\subsection}[runin]{}{\thetitle.}{0.5em}{\bfseries}[.]
\titleformat{name=\subsubsection}[runin]{}{\thetitle.}{0.5em}{\itshape}[.]
\titleformat{name=\paragraph,numberless}[runin]{}{}{0em}{}[.]
\titlespacing{\paragraph}{0em}{0em}{0.5em}
\titleformat{name=\subparagraph,numberless}[runin]{}{}{0em}{}[.]
\titlespacing{\subparagraph}{0em}{0em}{0.5em}

\providecommand{\keywords}[1]
{\small	
  \textit{Keywords:} #1
}
\providecommand{\MSC}[1]
{\small	
  \textit{2000 MSC:} #1
}

\counterwithin{figure}{section}
\counterwithin{table}{section}
\counterwithin{equation}{section} 

\theoremstyle{plain}
\newtheorem{theorem}{Theorem}[section] 
\newtheorem*{theorem*}{Theorem} 

\newtheorem*{definition*}{Definition}

\newtheorem*{lemma*}{Lemma}
\newtheorem{corollary}[theorem]{Corollary}
\newtheorem*{corollary*}{Corollary}
\newtheorem{remark}[theorem]{Remark}
\newtheorem*{remark*}{Remark}
\newtheorem{example}[theorem]{Example}
\newtheorem*{example*}{Example}

\newcommand{\cF}{\mathcal{F}}
\newcommand{\cP}{\mathcal{P}}
\newcommand{\cX}{\mathcal{X}}

\newcommand{\bN}{\mathbb{N}}
\newcommand{\bP}{\mathbb{P}}
\newcommand{\bR}{\mathbb{R}}
\newcommand{\bU}{\mathbb{U}}
\newcommand{\bV}{\mathbb{V}}
\newcommand{\bX}{\mathbb{X}}
\newcommand{\bY}{\mathbb{Y}}
\newcommand{\bZ}{\mathbb{Z}}
\newcommand{\xx}{\boldsymbol{x}}
\newcommand{\yy}{\boldsymbol{y}}
\newcommand{\zz}{\boldsymbol{z}}
\newcommand{\UU}{\boldsymbol{U}}

\newcommand{\1}{\mathds{1}}
\newcommand{\de}{\mathrm{\,d}}
\newcommand{\eqd}{\stackrel{\mathrm{d}}=}
\newcommand{\var}{\mathrm{Var}}
\newcommand{\rank}{\mathsf{rank}}

\begin{document}

\title{\Large\bf Hierarchical variable clustering based on the predictive strength between random vectors\footnote{This version: \today}}

\author{
\normalsize Sebastian Fuchs\footnote{Department for Artificial Intelligence \& Human Interfaces, University of Salzburg, Hellbrunner Strasse 34, 5020, Salzburg, Austria, {\tt sebastian.fuchs@plus.ac.at} } \\[1ex]
\and 
\normalsize Yuping Wang\footnote{Department for Artificial Intelligence \& Human Interfaces, University of Salzburg, Hellbrunner Strasse 34, 5020, Salzburg, Austria, {\tt yuping.wang@plus.ac.at}} \\ [1ex]
}

\date{}

\maketitle

\let\thefootnote\relax

\begin{abstract}
\noindent A rank-invariant clustering of variables is introduced that is based on the predictive strength between groups of variables, i.e., two groups are assigned a high similarity if the variables in the first group contain high predictive information about the behaviour of the variables in the other group and/or vice versa.
The method presented here is model-free, dependence-based and does not require any distributional assumptions.
Various general invariance and continuity properties are investigated, with special attention to those that are beneficial for the agglomerative hierarchical clustering procedure. A fully non-parametric estimator is considered whose excellent performance is demonstrated in several simulation studies and by means of real-data examples.
\end{abstract} 

\\
\MSC{62G05; 62H05; 62H20}

\section{Introduction}
\label{Sec:Intro}

A hierarchical variable cluster analysis aims at identifying interrelationships and associations present within a portfolio \(\cX = \{X_1, \dots, X_m\}\) of \(m \geq 3\) random variables and, thus, revealing important insights about the inner structure of the portfolio.
Variable clustering can therefore be particularly useful as a pre-processing step when analysing complex structured data in an unsupervised setting.
The key idea underlying this type of data analysis is to arrange \enquote{similar} objects (i.e., random variables) into groups (i.e., clusters) \cite{koch_analysis_2013}, though what is meant by \enquote{similarity} may have very different interpretations.

A common approach in hierarchical clustering so far is to consider similarity as a feature that occurs primarily between two individual random variables and that is detectable using measures of pairwise association such as (Pearson) correlation, Kendall's tau, Spearman's rho, or variants of these, in combination with some linkage function that assigns a degree of similarity also to groups of more than two random variables but takes into account exclusively pairwise information
(see, e.g., \cite{koch_analysis_2013, bonanno_2004} for clustering financial data, 
\cite{bonanno_2004, everitt_cluster_2011, fuchs_dissimilarity_2021, son_modified_2008} for clustering gene expression data, or \cite{bonanomi_defining_2017} for clustering preferences).
In this context, similarity is understood as a pairwise comovement of random variables, either linear or monotonic.
Alternative measures of bivariate similarity rely on tail dependence and tail correlation \cite{de_luca_tail_2011, DiLDurPap17RBN, DurPapTor14ADAC, fuchs_dissimilarity_2021} or on divergence measures such as mutual information \cite{banerjee_clustering_2005, emmert-streib_information_2008, jiang_clustering_2013, kojadinovic_agglomerative_2004, kraskov_hierarchical_2005, martinez_sotoca_supervised_2010, yang_novel_2019}.

The approaches mentioned above share a common drawback:
They all treat similarity as a pairwise feature and hence neglect possible higher-dimensional dependencies that can occur among more than two random variables.
Recently, in \cite{fuchs_dissimilarity_2021}, the authors solved this issue by introducing a general framework for quantifying multivariate similarity enabling higher dimensional dependencies to be taken into account.
In particular, they proposed several similarity measures based on multivariate concordance, thus treating similarity as equivalent to comonotonicity (an extreme dependence concept studied, e.g., in \cite{Dhaetal02a,PucWan15}).
With the same objective, in \cite{gijbels2023}, the authors proposed $\Phi$-dependence for quantifying multivariate similarity, thus treating similarity as equivalent to the presence of a dependence structure that is singular with respect to the product of its marginals.

In contrast to \cite{fuchs_dissimilarity_2021,gijbels2023}, in the present paper, two (disjoint) subsets \(\bX, \bY \subseteq \cX\) of random variables are assigned high similarity if the variables in set \(\bX\) contain high predictive information about the behaviour of the variables in the other set \(\bY\) and/or vice versa.
Thus, if two groups of random variables are clustered, then a high degree of functional dependence occurs between the two groups; in other words: the variables in one of the groups can be used to predict the variables in the second group.
\\
For that purpose, a hierarchical clustering procedure is introduced that is based on measures of predictability for multi-outcome data. Such a measure (of directed dependence) assigns to every pair of subsets \((\bX, \bY)\) a value in $[0,1]$,
equals $0$ if and only if 
the variables within \(\bY\) are independent of those within \(\bX\), 
and 
equals $1$ if and only if the variables within \(\bY\) perfectly depend on the variables within \(\bX\), i.e., 
each \(Y_k \in \bY\) equals a measurable function of the variables within \(\bX\) almost surely.
\\
To the best of the authors' knowledge, so far the only two types of measure of predictability applicable to multi-outcome data have been introduced in \cite{deb2020} and \cite{ansari_simple_2023}, whereby we here focus on the method introduced in \cite{ansari_simple_2023} due to its numerous favourable properties, the latter being a natural extension of Chatterjee's famous \enquote{coefficient of correlation} introduced in \cite{azadkia_simple_2021, chatterjee2021}.
\\
As a consequence, similarity is treated here as equivalent to (mutual) perfect dependence, a dependence concept being less restrictive than comonotonicity used in \cite{fuchs_dissimilarity_2021} but at the same time not as comprehensive as the concept of singularity used in \cite{gijbels2023}.

The paper is organized as follows: 
Following and adapting the approach in \cite{fuchs_dissimilarity_2021} for quantifying multivariate similarity, in Section \ref{Sec:Similarity} we introduce so-called dissimilarity functions being capable of detecting perfect dependence (type \ref{tPD:D}) or mutual perfect dependence (type \ref{tMPD:D}) between two non-empty subsets of $\cX$. 
For each type of dissimilarity function, key properties are presented and (where possible) its relation to alternative notions of multivariate similarity available in the literature (see, e.g., \cite{kojadinovic_agglomerative_2004,gijbels2023,fuchs_dissimilarity_2021}) is discussed. 
\\
Evidently, the choice of the dissimilarity function in question is crucial for the clustering result;
different clustering outputs may be obtained from the same data set when applying different dissimilarity functions.
And these clustering results may then provide new insights into the nature of the data, such as unexpected substructures whose revelation may trigger a more advanced data analysis.
In view of the potential impact of the chosen dissimilarity function on the clustering result, it is worth investigating whether the former exhibits certain desirable properties like invariance and continuity (Section \ref{Sec:Prop.}).
The invariance properties indicate that the introduced dissimilarity functions are dependence-based, and continuity ensures that a certain level of noise presented in the data does not cause the final clustering result to deviate too much.
\\
In Section \ref{Sec:Est} estimators for the proposed type \ref{tPD:D} and \ref{tMPD:D} dissimilarity functions are presented and it is shown that these estimators are strongly consistent and can be computed in $O(n \log n)$ time.
\\
Finally, Section \ref{Sec:Data} contains several simulation studies and real data examples which are used (1) for evaluating the performance of the different candidates for type \ref{tPD:D} and type \ref{tMPD:D} dissimilarity functions,
(2) for comparing the agglomerative hierarchical variable clustering methods based on perfect dependence (type \ref{tPD:D}) and mutual perfect dependence (type \ref{tMPD:D}) to alternative clustering methods available in the literature, 
(3) for elaborating their resilience to noise as an application of the above-mentioned continuity property and (4) for evaluating their performance when the number of random variables to be clustered is large.

Throughout this manuscript,
let \(m \geq 3\) be an integer which will be kept fixed 
and consider a (finite) set/portfolio \(\cX=\{X_1,\dots,X_m\}\) of random variables defined on the same probability space \((\Omega, \cF, \bP)\) which are always assumed to be non-degenerate, i.e., for all $k \in \{1,\dots,m\}$ the distribution of \(X_k\) does not follow a one-point distribution. 
Any subset of \(\cX\) will be denoted by upper-case black-board letters, e.g., \(\bX\), 
and let \(\cP_0 (\cX)\) denote the set of all non--empty subsets of \(\cX\).
Given a subset \(\bX = \{X_1, \dots, X_p\} \in \cP_0 (\cX)\) with \(p \leq m\), 
we indicate by \(\vec{\bX}\) a vector representation of \(\bX\), i.e., a $p$--dimensional random vector whose coordinates are distinct elements from \(\bX\). Clearly, the vector representation of any \(\bX\) need not be unique. 
Further, let \(L^0(\bR^p)\) denote the space of all $p$-dimensional random vectors.

\section{Hierarchical variable clustering based on (mutual) perfect dependence}
\label{Sec:Similarity}

Before turning to the hierarchical variable clustering procedure, 
we shall first specify what exactly we mean by \enquote{similarity} of two non-empty subsets of $\cX$
(not necessarily equal in cardinality), and how this \enquote{similarity} is to be quantified (Subsection \ref{SubSec:MPD}).
We achieve this by following the approach in \cite{fuchs_dissimilarity_2021} and introducing a so-called \enquote{dissimilarity function} (Subsection \ref{SubSec:Meth}), although we take the term similarity to be much more general than used in \cite{fuchs_dissimilarity_2021}. 
More precisely, we introduce dissimilarity functions being capable of detecting perfect dependence (type \ref{tPD:D}) or mutual perfect dependence (type \ref{tMPD:D}) between two non-empty subsets of $\cX$ (Subsection \ref{SubSec:DF.AB})
and, for each type of dissimilarity function, key properties are presented.
In addition, alternative concepts of multivariate similarity available in the literature are discussed and (where possible) their relation to (mutual) perfect dependence is elaborated
(Subsection \ref{SubSec:Comparison}).

\subsection{Quantifying (mutual) perfect dependence}
\label{SubSec:MPD}

We first introduce two dependence concepts for sets of random variables that are of key importance for this work (cf. \cite{lancaster1963}):
\begin{itemize}[leftmargin=1.5em]
\item[] \emph{Perfect dependence:}
A set \(\bY = \{Y_1,\dots,Y_q\} \in \cP_0(\cX)\), \(q \geq 1\), is said to be \emph{perfectly dependent} on the set \(\bX = \{X_1,\dots,X_p\} \in \cP_0(\cX)\), \(p \geq 1\),  
if, for every \(k \in \{1,\dots,q\}\), 
there exists some measurable function $f_k$ such that $Y_k = f_k (\vec{\bX})$ almost surely.

\item[] \emph{Mutual perfect dependence:}
Two sets \(\bX \in \cP_0(\cX)\) and \(\bY \in \cP_0(\cX)\)
are said to be \emph{mutually perfectly dependent} if \(\bY\) is perfectly dependent on \(\bX\) \emph{and} \(\bX\) is perfectly dependent on \(\bY\).
\end{itemize}
Perfect dependence of \(\bY\) on \(\bX\) is a concept of directed dependence in which the variables within the set \(\bX\) provide full information about the variables within \(\bY\), 
irrespective of the order of variables within $\bX$ and irrespective of the order of variables within $\bY$.
Apparently, mutual perfect dependence implies perfect dependence.

To uncover dependence structures with high predictive information, in what follows we employ 
the so-called \emph{simple extension of Azadkia \& Chatterjee's rank correlation} \(T^q\) introduced in \cite{ansari_simple_2023}: 
For the $q$-dimensional random vector $\vec{\bY}$, \(q \in \mathbb{N}\), given the $p$-dimensional random vector $\vec{\bX}$, \(T^q\) is defined by
\begin{align}
  T^q (\vec{\bY}|\vec{\bX})
  & := 1 - \frac{q - \sum_{i=1}^{q} T(Y_i | (\vec{\bX},Y_{i-1},\dots,Y_{1}))}{q - \sum_{i=1}^{q} T(Y_i | (Y_{i-1},\dots,Y_{1}))} 
  \label{Tq}
\end{align}
with \(T(Y_1|\emptyset) :=0 \)
and \(T\) denoting the \emph{simple measure of conditional dependence} introduced in \cite{azadkia_simple_2021} and given by
\begin{align}
  T(Y|\vec{\bX}) 
  & := \frac{\int_{\bR} \var\big(\bP(Y\geq y|\vec{\bX}) \big) \de \bP^Y(y)}{\int_{\bR} \var\left(\1_{\{Y\geq y\}}\right) \de \bP^Y(y)} 
  \label{T}
  \,.
\end{align}
\(T\) as defined in \eqref{T} fulfills 
\begin{align*}
    T(Y|\vec{\bX})
    & = \frac{\int_{\bR^p \times \bR^p} \int_\bR \Big[\bP(Y\geq y |\vec{\bX}=\xx_1)-\bP(Y\geq y |\vec{\bX}=\xx_2)\Big]^2 \de \bP^Y(y) \de \bP^{\vec{\bX}} \otimes \bP^{\vec{\bX}} (\xx_1,\xx_2)}{2 \int_\bR \var\left(\1_{\{Y\geq y\}}\right) \de \bP^Y(y)}
\end{align*}
and can therefore be considered either as a coefficient indicating how sensitive the conditional distribution of \(Y\) given \(\vec{\bX} = \xx\) is on \(\xx\) or, according to \eqref{T}, as an index quantifying the variability of the conditional distributions; see \cite{ansari_sens_2023}.

As shown in \cite{ansari_simple_2023}, \(T^q (\vec{\bY}|\vec{\bX})\) is invariant with respect to permutations of the variables within \(\bX = \{X_1,\dots,X_p\}\), and permutation invariance with respect to the variables within \(\bY = \{Y_1, \dots, Y_q\}\) can be achieved by defining the map
\begin{align}\label{Def.MoP}
  \kappa^{q|p} (\vec{\bY}|\vec{\bX}) 
	:= \frac{1}{q!} \sum_{\sigma\in S_q} T^q(\sigma(\vec{\bY})|\vec{\bX})
\end{align}
where \(\sigma(\vec{\bY}):= (Y_{\sigma_1},\ldots,Y_{\sigma_q})\) for \(\sigma=(\sigma_1,\ldots,\sigma_q)\in S_q\,,\)
the set of permutations of \(\{1,\ldots,q\}\).
$\kappa^{q|p}$ is a \emph{measure of predictability} for the $q$-dimensional random vector $\vec{\bY}$ given the $p$-dimensional random vector $\vec{\bX}$, i.e., 
it fulfills the following properties; see, e.g., \cite[Theorem 2.1, Corollary 2.7]{ansari_simple_2023}:
\begin{enumerate}[label=(A\arabic*), leftmargin=2.5em]
\item\label{MoP.A1} $0 \leq \kappa^{q|p}(\vec{\bY}|\vec{\bX}) \leq 1$.
\item\label{MoP.A2} $\kappa^{q|p}(\vec{\bY}|\vec{\bX}) = 0$ if and only if the variables in $\bY$ are independent of those in $\bX$.
\item\label{MoP.A3} $\kappa^{q|p}(\vec{\bY}|\vec{\bX}) = 1$ if and only if $\bY$ is perfectly dependent on $\bX$.
\end{enumerate}
In addition to the aforementioned three characteristics \ref{MoP.A1} to \ref{MoP.A3} which are crucial, 
\(\kappa^{q|p}\) fulfills a number of additional desirable properties including the information gain inequality, the characterization of conditional independence, various invariance properties, a monotonicity and a continuity property, the data processing inequality, and the self-equitability; we refer to \cite{ansari_simple_2023} for more background on these characteristics.
Several of these properties can be transferred to the corresponding dissimilarity functions and are studied in Section \ref{Sec:Prop.} below.

\begin{remark}~~
In \cite{deb2020}, the authors introduced the so-called \emph{kernel partial correlation (KPC) coefficient}, the so far only other known measure of predictability applicable to multi-outcome data, i.e., a measure satisfying \ref{MoP.A1} to \ref{MoP.A3}.
KPC measures the maximum mean discrepancy distance between the conditional distributions \(\bP^{\vec{\bY}|\vec{\bX} = \xx}\) and the unconditional distribution \(\bP^{\vec{\bY}}\),
and hence provides information on how much information on \(\vec{\bY}\) is gained by knowing \(\vec{\bX}\).
To the best of the authors' knowledge, only few properties of the KPC coefficient are known so far \cite[Proposition 8]{deb2020}, which is why in the present article we focus on \(\kappa^{q|p}\).
\end{remark}

\subsection{Dissimilarity functions based on (mutual) perfect dependence}
\label{SubSec:Meth}

For examining and performing a cluster analysis based on (mutual) perfect dependence,
we require the notion of a dissimilarity function:
A $(p,q)$-\emph{dissimilarity function} \(d^{p,q}\) based on (mutual) perfect dependence is a mapping that assigns to every pair 
\((\vec{\bX},\vec{\bY}) \in L^0(\bR^p) \times L^0(\bR^q)\) a non-negative value in \([0,\infty)\) with the following properties: 
\begin{enumerate}[label=(\(d\)\arabic*)]
\item \label{tSym:D}
Symmetry and permutation invariance:
The identities
\begin{align*}
    d^{p,q} (\vec{\bX},\vec{\bY}) 
    & = d^{q,p} (\vec{\bY},\vec{\bX}) 
    &   d^{p,q} (\vec{\bX},\vec{\bY}) 
    & = d^{p,q}(\sigma_{p} (\vec{\bX}),\sigma_{q}(\vec{\bY}))
\end{align*} 
hold for all \(\sigma_{p} \in S_{p}\) and  \(\sigma_{q} \in S_{q}\).

\item \label{tLI:D}
Law-invariance: 
\(d^{p,q} (\vec{\bX},\vec{\bY}) = d^{p,q} (\vec{\bX}_1,\vec{\bY}_1)\)
for all \((\vec{\bX},\vec{\bY}), (\vec{\bX}_1,\vec{\bY}_1) \in L^0(\bR^p) \times L^0(\bR^q) \) such that $(\vec{\bX},\vec{\bY}) \stackrel{d}{=} (\vec{\bX}_1,\vec{\bY}_1)$.

\item \label{tSim:D}
Multivariate similarity 
  \begin{enumerate}[label=(\Alph*)]
  \item \label{tPD:D} 
  based on perfect dependence: \\
  \(d^{p,q} (\vec{\bX},\vec{\bY}) = 0\) if and only if either \(\bY\) is perfectly dependent on \(\bX\) \emph{or} \(\bX\) is perfectly dependent on \(\bY\).
  \item \label{tMPD:D}
  based on mutual perfect dependence: \\
  \(d^{p,q} (\vec{\bX},\vec{\bY}) = 0\) if and only if both \(\bY\) is perfectly dependent on \(\bX\) \emph{and} \(\bX\) is perfectly dependent on \(\bY\).
  \end{enumerate}
\end{enumerate}

\ref{tSym:D} expresses a natural symmetry property between and within random vectors, i.e., the dissimilarity degree between two random vectors does not change when permuting the variables within each vector or when interchanging the random vectors.
In a nutshell, the dissimilarity degree \emph{does} depend on the information contained in the variables of a random vector, but not on how the variables are arranged (so it is a property of sets, not vectors).
The property ensures that the definition of a dissimilarity function is set compatible.
\ref{tLI:D} states that the dissimilarity function is law-invariant; see \cite{fuchs_dissimilarity_2021} for a more detailed discussion of these properties.

\ref{tSim:D}, instead, specifies what we understand by similarity between two random vectors: the occurrence of a directed functional relationship between the involved random vectors, either in one direction only \ref{tPD:D} or in both directions \ref{tMPD:D}.
In Subsection \ref{SubSec:Comparison} below we discuss and briefly summarize alternative notions of multivariate similarity available in the literature and (if possible) point out their relation to (mutual) perfect dependence.

\begin{remark}~~\label{Rem:DF.Ind}
Bearing in mind that independence represents the dependence structure with the least predictive information, an additional requirement to a dissimilarity function based on (mutual) perfect dependence could be to assign the maximal value to \(d^{p,q} (\vec{\bX},\vec{\bY})\) in the case \(\vec{\bX}\) and \(\vec{\bY}\) are independent. We will elaborate on this aspect below in Remark \ref{Diss.MOP:Rem.Ind}.
\end{remark}

Following \cite{fuchs_dissimilarity_2021}, all the dissimilarity functions can be glued together into the following concept which allows to assign a degree of dissimilarity to any pair of random vectors, regardless of the respective dimension:
An \emph{extended dissimilarity function} (of degree $m$) is a map
\begin{align*}
  d: \bigcup_{2 \leq p+q \leq m} L^0(\bR^p) \times L^0(\bR^q) \to [0,\infty)
\end{align*}
whose restriction
$d^{p,q} := d|_{L^0(\bR^p) \times L^0(\bR^q)}$ to $L^0(\bR^p) \times L^0(\bR^q)$ is a $(p,q)$-\emph{dissimilarity function} for all $ p,q \in \mathbb{N} $ with $ 2 \le p+q \le m$.

Below, the different steps of an agglomerative hierarchical clustering algorithm based on an extended dissimilarity function \(d\) are presented.
Recall that the overall objective is to determine a suitable partition of a given (finite) set \(\cX = \{X_1,\dots,X_m\}\) of \(m \geq 3\) random variables into non-empty and non-overlapping classes.
\begin{enumerate}[label=(Step \arabic*), leftmargin=4em]
\item 
Each random variable \(X_{k} \in \cX\) is assigned to its own class \(\bX_{k}\), \(k \in \{1,\dots,m\}\).

\item 
For each pair of classes \(\bX_1\) and \(\bX_2\), a degree of dissimilarity is calculated and recorded in a dissimilarity matrix. 

\item 
A pair of classes exhibiting the smallest degree of dissimilarity (i.e., being most similar) is identified and merged,
and the number of classes is reduced by one.

\item Steps 2 to 4 are repeated until the number of classes equals \(1\).
\end{enumerate}

\subsection{Type \ref{tPD:D} and \ref{tMPD:D} dissimilarity functions}
\label{SubSec:DF.AB}

We introduce a variety of type \ref{tPD:D} and type \ref{tMPD:D} dissimilarity functions, all of which are combinations of the coefficient \(\kappa\) defined in \eqref{Def.MoP} and a given merging function that aggregates the individual degrees of predictability for the two directions.
For each type, we recommend a suitable \enquote{aggregating} function (Remark \ref{Rem:DF.Recommend}) and present key properties of the resulting dissimilarity functions (Theorem \ref{MainThmDF}).

Inspired by correlation distances \cite{koch_analysis_2013}, i.e., 
Kendall' tau distance and Spearman's rho distance, 
we aim at constructing dissimilarity functions based on the measure of predictability \(\kappa\).
For \(p,q \geq 1\) and a function \(A: [0,1]^2 \to [0,\infty)\), 
we define the mapping \(d_A^{p,q}: L^0(\bR^p) \times L^0(\bR^q) \to [0,\infty)\) by
\begin{align} \label{Diss.MOP:Represent}
  d_A^{p,q} (\vec{\bX},\vec{\bY})
  := A \big(1-\kappa^{p|q}(\vec{\bX}|\vec{\bY}), 1-\kappa^{q|p}(\vec{\bY}|\vec{\bX})\big).
\end{align}
The mapping \(A\) \enquote{aggregates} the individual degrees of predictability in a desirable way.

\begin{corollary}\label{Diss.MOP:DF}
\(d_A^{p,q}\) is a \((p,q)\)-dissimilarity function if and only if \(A\) is symmetric and
\begin{itemize}
\item[{\rm \ref{tPD:D}}] \(A(0,v) = 0 = A(u,0)\) with \(A(u,v)>0\) for all \((u,v) \in (0,1]^2\).
\item[{\rm \ref{tMPD:D}}] \(A(0,0) = 0\) with \(A(u,v)>0\) for all \((u,v) \in [0,1]^2 \backslash \{(0,0)\}\).
\end{itemize}
\end{corollary}
\begin{proof}
Since \(\kappa^{q|p}\) is invariant under permutations within each vector, property \ref{tSym:D} is equivalent to $A$ being symmetric.
Further, property \ref{tLI:D} is immediate from the fact that \(\kappa^{q|p}\) is law invariant.
Property \ref{tSim:D}, instead, follows from the fact that \(\kappa^{q|p}\) characterizes perfect dependence.
\end{proof}

As functions \(A\) for aggregating the individual degrees of predictability, 
Corollary \ref{Diss.MOP:DF} justifies the use of bivariate copulas; for more background on copulas we refer to \cite{durante_principles_2015,Joe15,nelsen_introduction_2007}.
\begin{enumerate}
\item[{\rm \ref{tPD:D}}]  
Suitable candidates for \(A\) in the case of dissimilarity functions based on perfect dependence involve copulas \(C\) being symmetric and strictly positive on \((0,1]^2\),
so that \eqref{Diss.MOP:Represent} reads
\begin{align} \label{Diss.MOP:Represent.Cop}
  d_C^{p,q} (\vec{\bX},\vec{\bY})
  := C \big(1-\kappa^{p|q}(\vec{\bX}|\vec{\bY}), 1-\kappa^{q|p}(\vec{\bY}|\vec{\bX})\big).
\end{align}
This includes copulas from most Archimedean families like Frank, Gumbel, Joe and Clayton copulas (the latter with positive parameter only) but also elliptical copulas like Gaussian copulas and Student \(t\) copulas.
In particular, the independence copula \(\Pi\) and the Fr{\'e}chet-Hoeffding upper bound \(M\) generate dissimilarity functions based on perfect dependence, namely,
\begin{align*}
  d_{M}^{p,q} (\vec{\bX},\vec{\bY})
  & = \min \big\{1-\kappa^{p|q}(\vec{\bX}|\vec{\bY}), 1-\kappa^{q|p}(\vec{\bY}|\vec{\bX})\big\},
  \\
  d_{\Pi}^{p,q} (\vec{\bX},\vec{\bY})
  & = (1-\kappa^{p|q}(\vec{\bX}|\vec{\bY})) \, (1-\kappa^{q|p}(\vec{\bY}|\vec{\bX})).
\end{align*}
According to Corollary \ref{Diss.MOP:DF}, the Fr{\'e}chet-Hoeffding lower bound \(W\) fails to generate a proper dissimilarity function. 
Figure \ref{fig:funcDD} depicts the graph of function $d_{C}^{p,q}$ for various choices of \(C\).
\begin{figure}[htbp]
    \centering
    \includegraphics[width=0.32\textwidth]{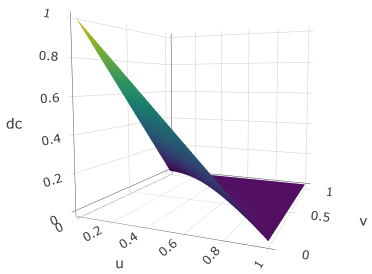} 
    \includegraphics[width=0.32\textwidth]{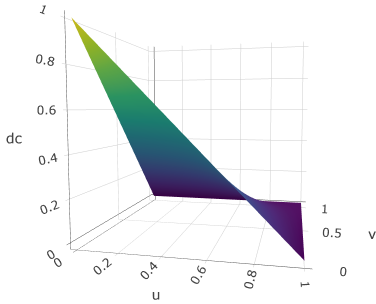} 
    \includegraphics[width=0.32\textwidth]{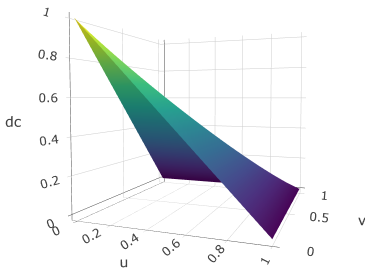} 
    \caption{The copula-based dissimilarity functions based on perfect dependence, where $C$ is the Gaussian copula with Kendall's tau $\tau \in\{-0.9, 0, 0.9\}$ (panels by column).}
    \label{fig:funcDD}
\end{figure}

\item[{\rm \ref{tMPD:D}}]
Suitable candidates for \(A\) in the case of dissimilarity functions based on mutual perfect dependence involve the dual of a symmetric copula.
The dual of a copula \(C\) is the function \(C^\ast: [0,1]^2 \to [0,1]\) given by \(C^\ast(u,v) := 1 - C(1-u,1-v)\), so that \eqref{Diss.MOP:Represent} then reads
\begin{align} \label{Diss.MOP:Represent.DCop}
  d_{C^\ast}^{p,q} (\vec{\bX},\vec{\bY})
  := 1- C \big(\kappa^{p|q}(\vec{\bX}|\vec{\bY}), \kappa^{q|p}(\vec{\bY}|\vec{\bX})\big).
\end{align}
In particular, duals of the Fr{\'e}chet-Hoeffding bounds \(W\) and \(M\) and the independence copula \(\Pi\) generate dissimilarity functions based on mutual perfect dependence, namely,
\begin{align*}
  d_{M^\ast}^{p,q} (\vec{\bX},\vec{\bY})
  & = 1 - \min \big\{\kappa^{p|q}(\vec{\bX}|\vec{\bY}), \kappa^{q|p}(\vec{\bY}|\vec{\bX})\big\},
  \\
  d_{\Pi^\ast}^{p,q} (\vec{\bX},\vec{\bY})
  & = 1 - \kappa^{p|q}(\vec{\bX}|\vec{\bY}) \, \kappa^{q|p}(\vec{\bY}|\vec{\bX}),
  \\
  d_{W^\ast}^{p,q} (\vec{\bX},\vec{\bY})
  & = 1 - \max \big\{\kappa^{p|q}(\vec{\bX}|\vec{\bY})+\kappa^{q|p}(\vec{\bY}|\vec{\bX})-1,0\big\}.
\end{align*}
Another possible construction principle is to average the individual degrees of predictability appropriately, i.e.,
\begin{align} \label{Diss.MOP:Represent.DCopAve}
  d_{\rm ave}^{p,q} (\vec{\bX},\vec{\bY})
  & = 1 - \frac{\kappa^{p|q}(\vec{\bX}|\vec{\bY}) + \kappa^{q|p}(\vec{\bY}|\vec{\bX})}{2}
\end{align}
which mimics but simplifies \(d_{W^\ast}^{p,q}\). 
Figure \ref{fig:funcMD} depicts the graph of function $d_{C^{*}}^{p,q}$ for various choices of \(C\).
\begin{figure}[htbp]
    \centering
    \includegraphics[width=0.32\textwidth]{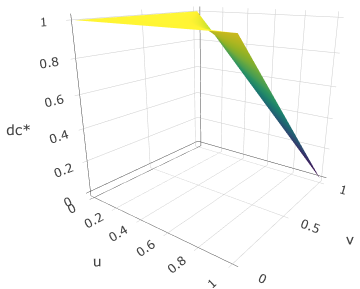} 
    \includegraphics[width=0.32\textwidth]{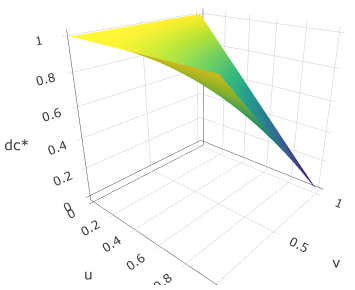} 
    \includegraphics[width=0.32\textwidth]{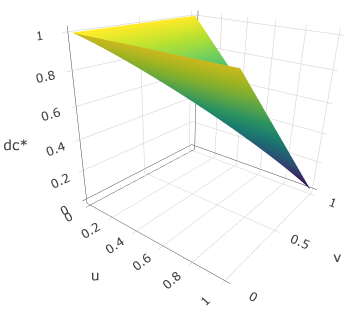} 
    \caption{The copula-based dissimilarity functions based on mutual perfect dependence, where $C$ is the Gaussian copula with Kendall's tau $\tau \in\{-0.9, 0, 0.9\}$ (panels by column).}
    \label{fig:funcMD}
\end{figure}
\end{enumerate}

We resume the discussion initiated in Remark \ref{Rem:DF.Ind}.

\begin{remark}[Behaviour of \(d^{p,q}\) when \(\vec{\bX}\) and \(\vec{\bY}\) are independent]~~\label{Diss.MOP:Rem.Ind} \\
Due to the properties of a measure of predictability, 
we have \(\kappa^{p|q}(\vec{\bX}|\vec{\bY}) = 0\) if and only if 
the variables within $\bY$ and those within $\bX$ are independent which in turn is equivalent to \(\kappa^{q|p}(\vec{\bY}|\vec{\bX})=0\).
Otherwise there exist \(\vec{\bX}\) and \(\vec{\bY}\) with \(\kappa^{p|q}(\vec{\bX}|\vec{\bY})\) being arbitrarily small but \(\kappa^{q|p}(\vec{\bY}|\vec{\bX})\) being close to \(1\) (see Example \ref{Ex.Asymm.Dep}).
From the perspective of predictive strength (and having in mind that independence represents the dependence structure with the least predictive information), it is therefore advisable that a dissimilarity function (type \ref{tPD:D} and \ref{tMPD:D}) equals \(1\) exclusively in the case of independence.\pagebreak
\\
For type \ref{tPD:D}, we observe that \(d_C^{p,q} (\vec{\bX},\vec{\bY})\) in \ref{Diss.MOP:Represent.Cop} equals \(1\) if and only if the variables in $\bY$ are independent of those in $\bX$.
Thus, using copulas for constructing dissimilarity functions of type \ref{tPD:D} is not only justified from the perspective of similarity (here: perfect dependence) but also from the perspective of independence.
\\
In order to achieve such a characterization of independence also for dissimilarity functions of type \ref{tMPD:D}, it is additionally to be required that \(A(u,v) = 1 \) if and only if \(u=1=v\). This excludes duals of copulas for the construction and draws particular attention to $d_{\rm ave}^{p,q}$.
\end{remark}

\begin{example}[Strongly asymmetric directed dependence]~~\label{Ex.Asymm.Dep} \\
Consider the random variables $X_1 \sim \mathcal{U}(0,1)$ and $X_2 = k \cdot X_1 \textrm{ modulus } 1$, $k \in \mathbb{N}$, whose joint distribution is depicted in Figure \ref{fig:scatter_ex}.
\begin{figure}[htbp] 
    \centering
    \subfigure[k=$3$]{
        \begin{minipage}[b]{0.47\textwidth}
        \centering
            \includegraphics[width=0.64\textwidth]{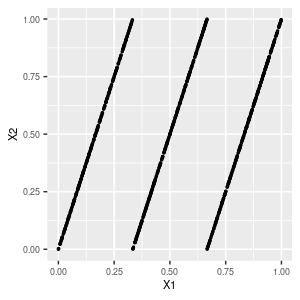} 
        \end{minipage}
    }
    \subfigure[k=$5$]{
        \begin{minipage}[b]{0.47\textwidth}
        \centering
            \includegraphics[width=0.64\textwidth]{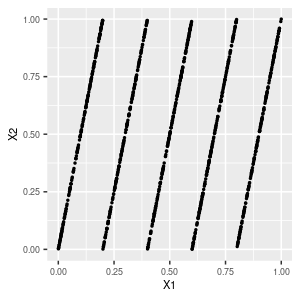} 
        \end{minipage}
    }
    \caption{Scatterplots of random variables $X_1$ and $X_2$ discussed in Example \ref{Ex.Asymm.Dep}}
    \label{fig:scatter_ex}
\end{figure}
Then \(X_2\) perfectly depends on \(X_1\) (but the variables are not mutually perfectly dependent) and straightforward but tedious calculation yields
\begin{align*}
  \kappa^{1,1}(X_1|X_2) 
  & = \frac{1}{k^2}
  & \kappa^{1,1}(X_2|X_1) 
  & = 1\,,
\end{align*}
and hence
$d_{\Pi}^{1,1} (X_1,X_2) = 0 $ and $d_{\mathrm{ave}}^{1,1} (X_1,X_2) = \frac{1}{2}-\frac{1}{2 k^2}$.
\end{example}

Building upon the previous discussion and the simulation study presented in Subsection \ref{Sim:DF.AB} we come up with the following recommendation for suitable dissimilarity functions of type \ref{tPD:D} and type \ref{tMPD:D}.

\begin{remark}[Recommendation for suitable type \ref{tPD:D} and \ref{tMPD:D} dissimilarity functions]~~\label{Rem:DF.Recommend}
\begin{enumerate}
\item[{\rm \ref{tPD:D}}] 
In Subsection \ref{Sim:DF.AB} we evaluate the different candidates for a dissimilarity function of type \ref{tPD:D} and observe that choosing copulas that are too close to the Fr{\'e}chet-Hoeffding upper bound \(M\) exhibits a tendency for the clustering output to produce chains (see Figure \ref{fig:dissfunc.choice.A}).
Therefore, we suggest using the independence copula \(\Pi\) as aggregating function and hence the dissimilarity function \(d^{p,q}_\Pi\).    

\item[{\rm \ref{tMPD:D}}]
In Subsection \ref{Sim:DF.AB} we evaluate the different candidates for a dissimilarity function of type \ref{tMPD:D} and observe that choosing copulas that are too close to the Fr{\'e}chet-Hoeffding lower bound \(W\) exhibits poor performance (see Figure \ref{fig:dissfunc.choice.B.2}).
Therefore and in light of the discussion in Remark \ref{Diss.MOP:Rem.Ind}, we suggest using the dissimilarity function \(d^{p,q}_{\rm ave}\).   
\end{enumerate}
Figure \ref{fig:disfunc} depicts the recommended type \ref{tPD:D} and \ref{tMPD:D} dissimilarity functions.
\begin{figure}[htbp] 
    \centering
    \subfigure
    {
        \begin{minipage}[b]{0.47\textwidth}
        \centering
            \includegraphics[width=0.85\textwidth]{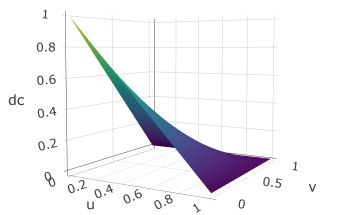} 
        \end{minipage}
    }
    \subfigure
    {
        \begin{minipage}[b]{0.47\textwidth}
        \centering
            \includegraphics[width=0.85\textwidth]{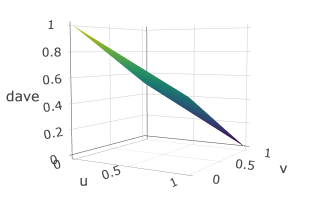} 
        \end{minipage}
    }
    \caption{Graphs of the recommended copula-based dissimilarity functions $d^{p,q}_{\Pi}$ of type \ref{tPD:D} (left panel) and $d^{p,q}_{\mathrm{ave}}$ of type \ref{tMPD:D} (right panel).}
    \label{fig:disfunc}
\end{figure}
\end{remark}

The next theorem encapsulates the findings from this section.

\begin{theorem}[Recommended type \ref{tPD:D} and \ref{tMPD:D} dissimilarity functions]~~\label{MainThmDF}
\begin{itemize}
\item[{\rm \ref{tPD:D}}]
The dissimilarity function $d^{p,q}_{\Pi}$ of type \ref{tPD:D} defined in \eqref{Diss.MOP:Represent.Cop} fulfills
\begin{enumerate}
    \item $d^{p,q}_{\Pi} (\vec{\bX},\vec{\bY}) \in [0,1]$.
    \item $d^{p,q}_{\Pi} (\vec{\bX},\vec{\bY}) = 0$ if and only if either \(\bY\) is perfectly dependent on \(\bX\) \emph{or} \(\bX\) is perfectly dependent on \(\bY\).
    \item $d^{p,q}_{\Pi} (\vec{\bX},\vec{\bY}) = 1$ if and only if 
the variables in $\bY$ are independent of those in $\bX$.
\end{enumerate}

\item[{\rm \ref{tMPD:D}}]
The dissimilarity function $d^{p,q}_{\mathrm{ave}}$ of type \ref{tMPD:D} defined in \eqref{Diss.MOP:Represent.DCopAve} fulfills
\begin{enumerate}
    \item $d^{p,q}_{\mathrm{ave}} (\vec{\bX},\vec{\bY}) \in [0,1]$.
    \item $d^{p,q}_{\mathrm{ave}} (\vec{\bX},\vec{\bY}) = 0$ if and only if \(\bY\) is perfectly dependent on \(\bX\) \emph{and} \(\bX\) is perfectly dependent on \(\bY\).
    \item $d^{p,q}_{\mathrm{ave}} (\vec{\bX},\vec{\bY}) = 1$ if and only if 
the variables in $\bY$ are independent of those in $\bX$.
\end{enumerate}
\end{itemize}
\end{theorem}

\subsection{Comparison with alternative hierarchical variable clustering methods}
\label{SubSec:Comparison}

Alternative notions of multivariate similarity available in the literature are discussed and (where possible) their relation to (mutual) perfect dependence is elaborated.
This results in a comparison of the recommended type \ref{tPD:D} and type \ref{tMPD:D} dissimilarity functions (Subsection \ref{SubSec:DF.AB}) with $\Phi$-dependence, dissimilarity functions based on measures of multivariate concordance and dissimilarity functions based on linkage methods.

\subsubsection*{Measures of multivariate concordance.}

In \cite{fuchs_dissimilarity_2021}, multivariate similarity is considered as equivalent to comonotonicity:
Two sets of random variables \(\bX \in \cP_0(\cX)\) and \(\bY \in \cP_0(\cX)\)
are said to be \emph{comonotonic} if all the random variables within \(\bX \cup \bY\) are pairwise comonotonic,
i.e., for each \(Z_1, Z_2 \in \bX \cup \bY\) there exists a random variable \(Z\) such that 
\((Z_1,Z_2) \stackrel{d}{=} (h_1(Z), h_2(Z))\) for some increasing functions \(h_1, h_2\).
Comonotonicity is a pairwise and symmetric dependence concept in the sense that two sets of random variables are comonotonic if and only if all the pairs within their union are comonotonic; see, e.g., \cite{puccetti2010}.
\\
Apparently, comonotonicity implies mutual perfect dependence which in turn implies perfect dependence.

We illustrate the differences between these two approaches also in terms of the clustering procedure by means of the following example.

\begin{example}[(Mutual) perfect dependence versus comonotonicity]~~\label{Ex.Comp.Con} \\
Consider the four dimensional random vector \((X_1,X_2,X_3,X_4)\) with 
\((X_1,X_2)\) being distributed according to the Fr{\'e}chet-Hoeffding lower bound, i.e., \((X_1,X_2) \sim W\), 
\((X_3,X_4)\) being distributed according to the Marshall-Olkin copula with parameter 
$(\alpha,\beta) = (1,0.5)$ (see, e.g., \cite{durante_principles_2015}), and 
such that the two vectors \((X_1,X_2)\) and \((X_3,X_4)\) are independent.
Figure \ref{fig:compareCC} depicts scatterplots of the two pairs of variables.
\begin{figure}[htbp] 
    \centering
    \subfigure[Scatterplot of $X_1$ and $X_2$]{
        \begin{minipage}[b]{0.47\textwidth}
        \centering
            \includegraphics[width=0.64\textwidth]{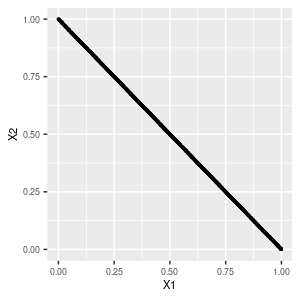} 
        \end{minipage}
    }
    \subfigure[Scatterplot of $X_3$ and $X_4$]{
        \begin{minipage}[b]{0.47\textwidth}
        \centering
            \includegraphics[width=0.64\textwidth]{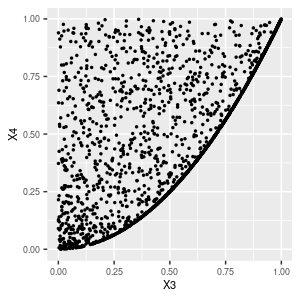} 
        \end{minipage}
    }
    \caption{Scatterplots of the pairs of variables discussed in Example \ref{Ex.Comp.Con}}
    \label{fig:compareCC}
\end{figure}
\\
According to \cite[Example 2, Example 4 and Theorem 4]{fuchs_quantifying_2022},
we obtain
\begin{align*}
    d_{\Pi}^{1,1} (X_1,X_2)
    & = (1-\kappa^{1,1}(X_1|X_2))(1-\kappa^{1,1}(X_2|X_1))
    = (1-1)^2
    = 0
    \\
    d_{\Pi}^{1,1} (X_3,X_4)
    & = (1-\kappa^{1,1}(X_3|X_4))(1-\kappa^{1,1}(X_4|X_3))
    = \left(1-\frac{2}{5}\right)\left(1-\frac{1}{3}\right)
    = \frac{2}{5}\,,
\end{align*}
and 
\begin{align*}
    d_{\mathrm{ave}}^{1,1} (X_1,X_2)
    & = 1-\frac{\kappa^{1,1}(X_1|X_2) + \kappa^{1,1}(X_2|X_1)}{2}
    = 1 - 1
    = 0
    \\
    d_{\mathrm{ave}}^{1,1} (X_3,X_4)
    & = 1-\frac{\kappa^{1,1}(X_3|X_4) + \kappa^{1,1}(X_4|X_3)}{2}
    = 1 - \frac{11}{30}
    = \frac{19}{30}\,.
\end{align*}
\\
Thus, the hierarchical clustering algorithms based on the recommended type \ref{tPD:D} and \ref{tMPD:D} dissimilarity functions first cluster \(X_1\) and \(X_2\) and then \(X_3\) and \(X_4\), i.e., both algorithms prefer the dependence structure of \((X_1,X_2)\) over that of \((X_3,X_4)\).
From the perspective of perfect dependence and mutual perfect dependence this is the correct choice as \(X_1\) contains more predictive information about \(X_2\) than \(X_3\) about \(X_4\) (and vice versa).
\\
Instead, the hierarchical clustering algorithms based on measures of concordance, such as those introduced in \cite[Section 3.2]{fuchs_dissimilarity_2021} that are related to Kendall's tau and Spearman's rho, clearly first cluster \(X_3\) and \(X_4\) and then \(X_1\) and \(X_2\), i.e., the algorithms prefer the less predictive but more concordant dependence structure of \((X_3,X_4)\) over that of \((X_1,X_2)\).
This is due to the fact that measures of concordance are monotone with respect to the pointwise/concordance order and that the copula of \(X_3\) and \(X_4\) pointwise exceeds the copula of \(X_1\) and \(X_2\). 
In other words: the dependence structure of \((X_1,X_2)\) is countermonotonic and hence less concordant than the dependence structure of \((X_3,X_4)\).    
\end{example}\pagebreak

\subsubsection*{$\Phi$-dependence.}

While restricting to an absolutely continuous setting, in \cite{gijbels2023} the authors introduce $\Phi$-dependence, 
a measure being capable of characterizing both the independence of random vectors (like \(\kappa\) does) and the singularity of the joint dependence structure with respect to the product of the marginal dependence structures.
Thus, in \cite{gijbels2023} multivariate similarity is considered as equivalent to the presence of a dependence structure that is singular with respect to the product of its marginals. This allows for the detection of tail dependence, in particular.
\\
Apparently, in such a setting perfect dependence implies singularity.
As a consequence, dependence structures such as given in \cite[Figure 1, right panel]{gijbels2023} are considered to be highly similar although the corresponding degree of predictability is rather low. 
We illustrate the differences between these two approaches also in terms of the clustering procedure by means of the following example.

\begin{example}[(Mutual) perfect dependence versus $\Phi$-dependence]~~\label{Ex.Comp.PhiDep} \\
Consider the four dimensional random vector \((X_1,X_2,X_3,X_4)\) with 
\((X_1,X_2)\) being distributed according to the average of the Fr{\'e}chet-Hoeffding upper and lower bound, i.e., \((X_1,X_2) \sim \tfrac{M+W}{2}\), 
\((X_3,X_4)\) being distributed according to an ordinal sum of \(\Pi\) (see, e.g., \cite{durante_principles_2015}) with respect to the partition \((0,0.3),(0.3,0.5),(0.5,0.75),(0.75,1)\), and such that the two vectors \((X_1,X_2)\) and \((X_3,X_4)\) are independent.
Figure \ref{fig:compareMI} depicts scatterplots of the two pairs of variables.
\begin{figure}[htbp] 
    \centering
    \subfigure[Scatterplot of $X_1$ and $X_2$]{
        \begin{minipage}[b]{0.47\textwidth}
        \centering
            \includegraphics[width=0.64\textwidth]{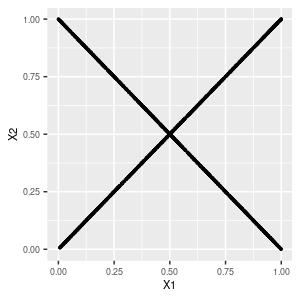} 
        \end{minipage}
    }
    \subfigure[Scatterplot of $X_3$ and $X_4$]{
        \begin{minipage}[b]{0.47\textwidth}
        \centering
            \includegraphics[width=0.64\textwidth]{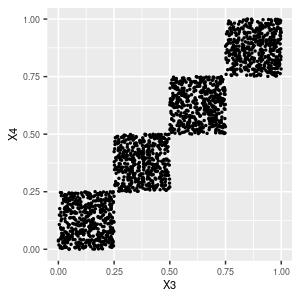} 
        \end{minipage}
    }
    \caption{Scatterplots of the pairs of variables discussed in Example \ref{Ex.Comp.PhiDep}}
    \label{fig:compareMI}
\end{figure}
\\
According to \cite[Example 2, Example 4 and Theorem 4]{fuchs_quantifying_2022} and due to fact that the variables within \((X_1,X_2)\) as well as within \((X_3,X_4)\) are exchangeable,
we have
\begin{align*}
    d_{\Pi}^{1,1} (X_1,X_2)
    & = (1-\kappa^{1,1}(X_1|X_2))(1-\kappa^{1,1}(X_2|X_1))
    = (1-\frac{1}{4})^2
    = \frac{9}{16}
    \\
    d_{\Pi}^{1,1} (X_3,X_4)
    & = (1-\kappa^{1,1}(X_3|X_4))(1-\kappa^{1,1}(X_4|X_3))
    = (1-\frac{3}{4})^2
    = \frac{1}{16}\,,
\end{align*}
and 
\begin{align*}
    d_{\mathrm{ave}}^{1,1} (X_1,X_2)
    & = 1-\frac{\kappa^{1,1}(X_1|X_2) + \kappa^{1,1}(X_2|X_1)}{2}
    = 1 - \frac{1}{4}
    = \frac{3}{4}
    \\
    d_{\mathrm{ave}}^{1,1} (X_3,X_4)
    & = 1-\frac{\kappa^{1,1}(X_3|X_4) + \kappa^{1,1}(X_4|X_3)}{2}
    = 1 - \frac{3}{4}
    = \frac{1}{4}\,.
\end{align*}

Thus, the hierarchical clustering algorithms based on the recommended type \ref{tPD:D} and \ref{tMPD:D} dissimilarity functions first cluster \(X_3\) and \(X_4\) and then \(X_1\) and \(X_2\), i.e., both algorithms prefer the dependence structure of \((X_3,X_4)\) over that of \((X_1,X_2)\).
From the perspective of perfect dependence and mutual perfect dependence this is the correct choice as \(X_3\) contains more predictive information about \(X_4\) than \(X_1\) about \(X_2\) (and vice versa).
\\
Instead, and according to \cite{gijbels2023}, the hierarchical clustering algorithm based on \(\Phi\)-dependence first clusters \(X_1\) and \(X_2\) and then \(X_3\) and \(X_4\), i.e., the algorithm prefers the less predictive but more singular dependence structure of \((X_1,X_2)\) over that of \((X_3,X_4)\).
This is due to the fact that the dependence structure of \((X_1,X_2)\) is singular and hence attains maximum \(\Phi\)-dependence.
\end{example}

\subsubsection*{Linkage methods.}

When aiming at reducing the computation time in hierarchical variable clustering, usually linkage methods come into play. A linkage method relates the dissimilarity degree between two sets of variables $\mathbb{X} = \{X_{1},\dots,X_{p}\}$ and $\mathbb{Y} = \{Y_{1},\dots,Y_{q}\}$ to the pairwise dissimilarity degrees $d^{1,1}(X_i,Y_j)$ of variables \(X_i\) and \(Y_j\) from both the sets \cite{koch_analysis_2013, everitt_cluster_2011}. 
The three most common linkage methods are
\begin{itemize}
    \item [1.] Single linkage: 
    \begin{align*}
        d_{\mathrm{single}}^{p,q}(\Vec{\mathbb{X}} ,\Vec{\mathbb{Y}}) 
        & \coloneqq \min_{ i\in \{1,\dots,p\},j\in \{1,\dots.q\}} d^{1,1}( X_{i} ,Y_{j})
    \end{align*}
    \item [2.] Average linkage: 
    \begin{align*}
        d_{\mathrm{average}}^{p,q}(\Vec{\mathbb{X}} ,\Vec{\mathbb{Y}}) 
        & \coloneqq \frac{1}{p \cdot q}\sum _{i=1}^{p}\sum _{j=1}^{q} d^{1,1}( X_{i} ,Y_{j})
    \end{align*}
    \item [3.] Complete linkage: 
    \begin{align*}
        d_{\mathrm{complete}}^{p,q}(\Vec{\mathbb{X}} ,\Vec{\mathbb{Y}}) 
        & \coloneqq \max_{ i\in \{1,\dots,p\},j\in \{1,\dots.q\}} d^{1,1}(X_{i} ,Y_{j})
    \end{align*}
\end{itemize}

Dissimilarity functions based on linkage methods exhibit many desirable properties, which in most cases are inherited from the underlying pairwise dissimilarity function $d^{1,1}$;
we refer to \cite{fuchs_dissimilarity_2021} for a detailed discussion of dissimilarity functions based on linkage methods.

Even though linkage methods offer some advantages, they all share the same structural drawback: 
They take into account solely pairwise information implying that the value of a \((p,q)\)-dissimilarity function of two sets of variables only depends on the pairwise interrelations between the two sets.
The following two examples illustrate this structural drawback also in terms of the clustering procedure:

\begin{example}[Perfect dependence versus linkage methods]~~\label{Ex.Comp.Link} \\
Consider the four dimensional random vector \((X_1, X_2, X_3, X_4)\) with 
$X_1 \sim \mathcal{N}(0,1)$, $X_2 \sim \mathcal{N}(0,1)$, 
$X_3 = \frac{X_1}{2} + X_2$ and $X_4 = X_2 + \epsilon$ where $\epsilon \sim \mathcal{N}(0,1)$.
Figure \ref{fig:ScatterplotLink} depicts the pairwise dependence structures of the four variables.
\begin{figure}[htbp]
    \centering
    \includegraphics[width=1\textwidth]{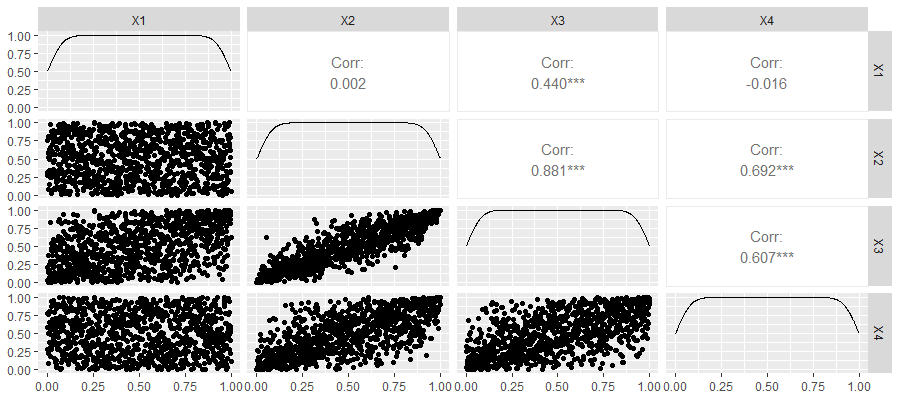}
    \caption{Pairwise dependence structure of the four random variables ($n = 1,000$) discussed in Example \ref{Ex.Comp.Link}}
    \label{fig:ScatterplotLink}
\end{figure}
To cluster the four variables, we first use the (multivariate) type \ref{tPD:D} dissimilarity function \(d^{p,q}_\Pi\) and then the (pairwise) type \ref{tPD:D} dissimilarity function \(d^{1,1}_\Pi\) in combination with the three linkage methods. The thereby obtained clustering results are presented in Figure \ref{fig:cluster_link}. 
\begin{figure}[htbp]
\centering
    \subfigure[Type \ref{tPD:D}]{
        \begin{minipage}[b]{0.47\textwidth}
        \centering
        \includegraphics[width=0.60\textwidth]{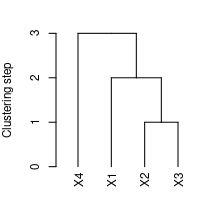} 
        \end{minipage}
    }
    \subfigure[Type \ref{tPD:D} + Single linkage]{
        \begin{minipage}[b]{0.47\textwidth}
        \centering
        \includegraphics[width=0.60\textwidth]{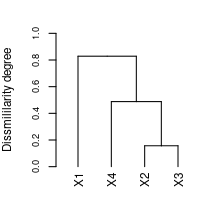} 
        \end{minipage}
    }
    \\
    \subfigure[Type \ref{tPD:D} + Average linkage]{
        \begin{minipage}[b]{0.47\textwidth}
        \centering
        \includegraphics[width=0.60\textwidth]{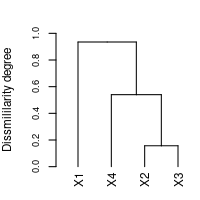} 
        \end{minipage}
    }
    \subfigure[Type \ref{tPD:D} + Complete linkage]{
        \begin{minipage}[b]{0.47\textwidth}
        \centering
        \includegraphics[width=0.60\textwidth]{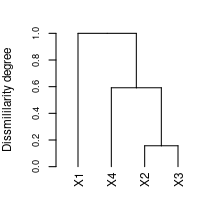} 
        \end{minipage}
    }
    \caption{Dendrograms of the hierarchical clustering in Example \ref{Ex.Comp.Link} when using type \ref{tPD:D} and type \ref{tPD:D} dissimilarity functions in combination with three linkage methods ($n=1,000$).}
    \label{fig:cluster_link}
\end{figure}
We observe that all methods first cluster variables \(X_2\) and \(X_3\).
However, in the second step the hierarchical clustering algorithms based on the (pairwise) dissimilarity functions of type \ref{tPD:D} in combination with any linkage method then clusters $X_4$ with $X_2$ and $X_3$ while the (multivariate) dissimilarity function of type \ref{tPD:D} clusters $X_1$ with $X_2$ and $X_3$.
We therefore find that even though the linkage method is more efficient in terms of computation time, 
it is not capable of recognising the perfect dependence of \(X_3\) on \((X_1,X_2)\), unlike the (multivariate) type \ref{tPD:D} dissimilarity function.
\end{example}\pagebreak

It even turns out that linkage methods fail to constitute proper (multivariate) type \ref{tPD:D} and \ref{tMPD:D} dissimilarity functions in the sense of Subsection \ref{SubSec:Meth}.

\begin{example}[Linkage methods fail to constitute type \ref{tPD:D} and \ref{tMPD:D} dissimilarity functions]~~\label{Ex:LinkageFail}
\begin{enumerate}
\item[{\rm \ref{tPD:D}}] 
Consider \(d^{1,1}_\Pi\), two independent random variables \(X_1\) and \(X_2\) and define \(X_3 := X_1 + X_2\). 
Then \(X_3\) is completely dependent on \((X_1,X_2)\),
however, 
\begin{align*}
  d^{1,1}_\Pi(X_1,X_3)  = (1-\underbrace{\kappa^{1|1}(X_1|X_3)}_{<1})(1-\underbrace{\kappa^{1|1} (X_3|X_1))}_{<1}
    > 0 
\end{align*}
as well as \(d^{1,1}_\Pi(X_2,X_3) > 0\) and hence 
\begin{align*}
    d_{\mathrm{single}}^{2,1}((X_1,X_2),X_3) 
    & > 0 = d^{2,1}_\Pi((X_1,X_2),X_3)\,, 
    \\
    d_{\mathrm{average}}^{2,1}((X_1,X_2),X_3)
    & > 0 = d^{2,1}_\Pi((X_1,X_2),X_3)\,, 
    \\
    d_{\mathrm{complete}}^{2,1}((X_1,X_2),X_3)
    & > 0 = d^{2,1}_\Pi((X_1,X_2),X_3)\,. 
\end{align*}
Therefore, linkage methods fail to constitute proper type \ref{tPD:D} dissimilarity functions.
\item[{\rm \ref{tMPD:D}}] 
Consider \(d^{1,1}_{\rm ave}\), two independent random variables \(X_1\) and \(X_2\) and define \(X_3 := X_1\) and \(X_4 := X_2\). 
Then \((X_1,X_2)\) and \((X_3,X_4)\) are mutually completely dependent,
however, 
\begin{align*}
    d_{\mathrm{average}}^{2,2}((X_1,X_2),(X_3,X_4))
    & = 0.5 > 0 = d^{2,2}_{\rm ave}((X_1,X_2),(X_3,X_4))\,,
    \\
    d_{\mathrm{complete}}^{2,2}((X_1,X_2),(X_3,X_4))
    & = 1 > 0 = d^{2,2}_{\rm ave}((X_1,X_2),(X_3,X_4))\,.
\end{align*}
In addition, \((X_1,X_2)\) and \(X_3\) fail to be mutually completely dependent,
however, 
\begin{align*}
    d_{\mathrm{single}}^{2,1}((X_1,X_2),X_3) 
    & = 0 < d^{2,1}_{\rm ave}((X_1,X_2),X_3)\,.
\end{align*}
Therefore, linkage methods even fail to constitute proper type \ref{tMPD:D} dissimilarity functions.
\end{enumerate}
\end{example}

\begin{remark}[Further alternative notions of multivariate similarity]~~\\
In \cite{fuchs_dissimilarity_2021}, the authors further introduced a dissimilarity function based on the multivariate tail dependence coefficient, hence considering similarity as equivalent to the occurrence of maximum lower tail dependence. 
\end{remark}

\section{Invariance and continuity}
\label{Sec:Prop.}

In this section a continuity property and various invariance properties (according to \cite{fuchs_dissimilarity_2021}) for the dissimilarity functions $d^{p,q}_{\Pi}$ of type \ref{tPD:D} and $d^{p,q}_{\mathrm{ave}}$ of type \ref{tMPD:D} are presented.
From the invariance properties we can conclude that the dissimilarity functions at hand are dependence-based, and continuity ensures that a certain level of noise present in the data does not cause the final clustering result to deviate too much.
We illustrate this resilience of \(d^{p,q}_{\Pi}\) and \(d^{p,q}_{\mathrm{ave}}\) to noise via a simulation study (see Subsection \ref{SubSec:Data:Noise} below).

\bigskip
Recall that \(\kappa^{q|p}\) defined in \ref{Def.MoP} is a measure of predictability, i.e., 
\(\kappa^{q|p}(\vec{\bY}|\vec{\bX}) \in [0,1]\), \(\kappa^{q|p}(\vec{\bY}|\vec{\bX})\) equals $0$ exclusively in the case the variables within $\bY$ are independent of those within $\bX$, and equals $1$ if and only if $\bY$ is perfectly dependent on $\bX$.
In addition and as mentioned in Section \ref{Sec:Similarity}, 
\(\kappa^{q|p}\) fulfills a number of additional desirable properties some of which can be transferred to the corresponding dissimilarity functions $d^{p,q}_{\Pi}$ of type \ref{tPD:D} and $d^{p,q}_{\mathrm{ave}}$ of type \ref{tMPD:D}.

As a first important property of type \ref{tPD:D} and \ref{tMPD:D} dissimilarity functions, 
we show their invariance when replacing the variables within $\bX$ and those within $\bY$ by their individual distributional transforms, i.e., the variables can be replaced by their individual ranks.
Corollary \ref{DF.Cor.DInv} is an immediate consequence of \cite[Corollaries 2.5 and 2.6]{ansari_simple_2023}.

\begin{corollary}[Distributional invariance]~~\label{DF.Cor.DInv} \\
Let \(\bU\) and \(\bV\) denote the sets of individual distributional transforms of \(\bX\) and \(\bY\), respectively, i.e.,
\(\bU = \{F_{X_1}(X_1), \dots, F_{X_p}(X_p)\}\) and \(\bV = \{F_{Y_1}(Y_1), \dots, F_{Y_q}(Y_q)\}\).
Then
\begin{align*}
  d^{p,q} (\vec{\bX},\vec{\bY})
  = d^{p,q} (\vec{\bU},\vec{\bV})\,,
\end{align*}
where \(d^{p,q}\) is either a dissimilarity function of type \ref{tPD:D} or \ref{tMPD:D} according to \eqref{Diss.MOP:Represent.Cop}, \eqref{Diss.MOP:Represent.DCop} or \eqref{Diss.MOP:Represent.DCopAve}. 
\end{corollary}

The next result is due to \cite[Theorem 2.10]{ansari_simple_2023} and demonstrates that the value of a dissimilarity function remains unchanged when transforming the variables within $\bX$ and those within $\bY$ by strictly increasing and bijective transformations.

\begin{corollary}[Invariance under strictly increasing bijective transformations]~~\label{DF.Cor.InvBijT}\\
Let \(g_i, h_k: \bR \mapsto \bR\), \(i \in \{1,\dots,p\}, k \in \{1,\dots,q\}\), 
be strictly increasing and bijective transformations, 
and let \(\bU\) and \(\bV\) denote the sets of transformed variables, i.e.,
\(\bU = \{g_1(X_1), \dots, g_p(X_p)\}\) and \(\bV = \{h_1(Y_1), \dots, h_q(Y_q)\}\).
Then
\begin{align*}
  d^{p,q} (\vec{\bX},\vec{\bY})
  = d^{p,q} (\vec{\bU},\vec{\bV})\,,
\end{align*}
where \(d^{p,q}\) is either a dissimilarity function of type \ref{tPD:D} or \ref{tMPD:D} according to \eqref{Diss.MOP:Represent.Cop}, \eqref{Diss.MOP:Represent.DCop} or \eqref{Diss.MOP:Represent.DCopAve}. 
\end{corollary}
\noindent In a nutshell, Corollary \ref{DF.Cor.InvBijT} confirms that
the concept of dissimilarity considered here is dependence-based.

\bigskip
Although, according to Theorem \ref{MainThmDF}, the values \(0\) and \(1\) of the dissimilarity functions \(d^{p,q}_{\Pi}\) and \(d^{p,q}_{\mathrm{ave}}\) have a clear interpretation,
the meaning of \(d^{p,q}_{\Pi}\) and \(d^{p,q}_{\mathrm{ave}}\) taking values in the interval \((0,1)\) is not specified. 
This justifies the investigation of modes of convergence that are compatible with the introduced dissimilarity functions.
\\
Since \(\kappa^{q|p}\) depends on conditional distributions, convergence in distribution of the sequence \((\vec{\bX}_n,\vec{\bY}_n)_{n\in \bN}\) is not sufficient for the convergence of \((d^{p,q}_{\Pi} (\vec{\bX},\vec{\bY}))_{n\in \bN}\) and \((d^{p,q}_{\mathrm{ave}} (\vec{\bX},\vec{\bY}))_{n\in \bN}\); see \cite[Section 4]{ansari_simple_2023}.
A promising candidate for achieving a continuity statement has been presented in \cite[Theorem 4.1 and Corollary 4.3]{ansari_simple_2023} where the authors showed that \(T^q\) (and consequently \(\kappa^{q|p}\)) is continuous with respect to the notion of conditional weak convergence going back to \cite{Sweeting_1989}.
Applying these results, we show that the here considered dissimilarity functions are continuous in classes of elliptical and $l_1$-norm symmetric distributions.

For \(d\in \bN\,,\) denote by \(\mathcal{U}(\bR^d)\) a class of bounded, continuous, weak convergence-determining functions mapping from \(\bR^d\) to \(\mathbb{C}\). 
Denote by \(\xrightarrow{ d }\) convergence in distribution. 
A sequence \((f_n)_{n\in \bN}\) of functions mapping from \(\bR^d\) to \(\mathbb{C}\) is said to be \emph{asymptotically equicontinuous} on an open set \(V\subset \bR^d\,,\) if for all \(\varepsilon>0\) and \(\xx \in V\) there exist \(\delta(\xx,\varepsilon)>0\) and \(n(\xx,\varepsilon)\in \bN\) such that whenever \(|\yy-\xx|\leq \delta(\xx,\varepsilon)\) then \(|f_n(\yy)-f_n(\xx)|<\varepsilon\) for all \(n > n(\xx,\varepsilon)\,.\) 
Further, \((f_n)_{n\in \bN}\) is said to be \emph{asymptotically uniformly equicontinuous} on \(V\) if it is asymptotically equicontinuous on \(V\) and the constants \(\delta(\varepsilon)=\delta(\xx,\varepsilon)\) and \(n(\varepsilon)=n(\xx,\varepsilon)\) do not depend on \(\xx\).

\begin{theorem}[Continuity, general result]~~\label{Thm:Cont} \\
Consider the $(p+q)$-dimensional random vector $(\vec{\bX},\vec{\bY})$ and a sequence $(\vec{\bX}_n,\vec{\bY}_n)_{n\in \mathbb{N}}$ of $(p+q)$-dimensional random vectors. 
Let $V_{1} \subset \mathbb{R}^{p}$ and $W_{1} \subset \mathbb{R}^{q}$ be open such that $\mathbb{P}(\vec{\bX} \in V_{1})=1$ and $\mathbb{P}(\vec{\bY} \in W_{1})=1$.
For each choice of permutations \(\sigma \in S_p\) and all $i\in \{2, \dots ,p\}$, 
let further $O_{i} \subset \mathbb{R}^{i-1}$ and $W_{i} \subset \mathbb{R}^{q+i-1}$ be open such that $\mathbb{P}((X_{\sigma_1},\dots,X_{\sigma_{i-1}}) \in O_{i}) =1$ and 
$\mathbb{P}((\vec{\bY},X_{\sigma_1},\dots,X_{\sigma_{i-1}}) \in W_{i})=1$,
and for each choice of permutations \(\tau \in S_q\) and all $j\in \{2,\dots ,q\}$, 
let $U_{j} \subset \mathbb{R}^{j-1}$ and $V_{j} \subset \mathbb{R}^{p+j-1}$ be open such that $\mathbb{P}((Y_{\tau_1} ,\dots,Y_{\tau_{j-1}}) \in U_{j})=1$ and 
$\mathbb{P}((\vec{\bX},Y_{\tau_1} ,\dots,Y_{\tau_{j-1}}) \in V_{j})=1$. 
If, further
\begin{itemize}
\item 
$(\vec{\bX}_{n}, \vec{\bY}_{n}) \xrightarrow{d} (\vec{\bX},\vec{\bY})$, 

\item 
for all \(u \in \mathcal{U}(\bR)\),
$(\mathbb{E}[u(X_{\sigma_{1},n}) | \vec{\bY}_{n} = \yy])_{n\in \mathbb{N}}$ is asymptotically equicontinuous on $W_{1}$ and, for all $i\in \{2, \dots ,p\}$, 
$(\mathbb{E}[u( X_{\sigma_i,n}) | (\vec{\bY}_{n}, X_{\sigma_1,n} ,\dots,X_{\sigma_{i-1},n}) = \zz])_{n\in \mathbb{N}}$ is asymptotically equicontinuous on $W_{i}$ and 
$(\mathbb{E}[ u(X_{\sigma_i,n}) | ( X_{\sigma_1,n},\dots,X_{\sigma_{i-1},n}) = \xx])_{n\in \mathbb{N}}$ is asymptotically equicontinuous on $O_{i}$, 

\item 
for all \(u \in \mathcal{U}(\bR)\), 
$(\mathbb{E}[u(Y_{\tau_{1},n}) | \vec{\bX}_{n} = \xx])_{n\in \mathbb{N}}$ is asymptotically equicontinuous on $V_{1}$ and, for all $j\in \{2,\dots ,q\}$, 
$(\mathbb{E}[u(Y_{\tau_j,n}) | (\vec{\bX}_{n},Y_{\tau_1,n}, \dots, Y_{\tau_{j-1},n}) = \zz])_{n\in \mathbb{N}}$ is asymptotically equicontinuous on $V_{j}$ and
$(\mathbb{E}[u(Y_{\tau_j,n}) | (Y_{\tau_1,n}, \dots, Y_{\tau_{j-1},n}) = \yy])_{n\in \mathbb{N}}$ is asymptotically equicontinuous on $U_{j}$, and

\item $F_{X_{i,n}} \circ F_{X_{i,n}}^{-1}( t)\xrightarrow{n\rightarrow \infty } F_{X_{i}} \circ F_{X_{i}}^{-1}( t)$ and $F_{Y_{j,n}} \circ F_{Y_{j,n}}^{-1}( t)\xrightarrow{n\rightarrow \infty } F_{Y_{j}} \circ F_{Y_{j}}^{-1}( t)$ for $\lambda $-almost all $t\in ( 0,1)$ and for all $i\in \{1, \dots ,p\}$ and $j\in \{1,\dots,q\}$, 
\end{itemize}
then $$d^{p,q}(\vec{\bX}_{n} ,\vec{\bY}_{n})\xrightarrow{n\rightarrow \infty } d^{p,q}(\vec{\bX} ,\vec{\bY})\,,$$ where \(d^{p,q}\) is either a dissimilarity function of type \ref{tPD:D} or \ref{tMPD:D} according to \eqref{Diss.MOP:Represent.Cop}, \eqref{Diss.MOP:Represent.DCop} or \eqref{Diss.MOP:Represent.DCopAve}. 
\end{theorem}
\begin{proof}
Continuity of \(d^{p,q}\) can be deduced from continuity of \(\kappa^{q|p}\) and \(\kappa^{p|q}\) (see \cite[Corollary 4.3]{ansari_simple_2023}) and the fact that every copula is continuous and the average of continuous function is also continuous.
\end{proof}

As a direct consequence of Theorem \ref{Thm:Cont}, we conclude that, for elliptical distributions, 
the dissimilarity functions \(d^{p,q}_{\Pi}\) and \(d^{p,q}_{\mathrm{ave}}\) are continuous in the scale matrix and the radial part: 
\\
A random vector \((\vec{\bX},\vec{\bY})\) is said to be \emph{elliptically distributed} for some vector \({\boldsymbol \mu}\in \bR^{p+q}\,,\) 
some positive semi-definite matrix \(\Sigma=(\sigma_{ij})_{1\leq i,j\leq p+q}\,,\) and some generator \(\phi\colon \bR_+ \to \bR\), \((\vec{\bX},\vec{\bY}) \sim \mathcal{E}({\boldsymbol \mu},\Sigma,\phi)\) for short, 
if the characteristic function of \((\vec{\bX},\vec{\bY})-{\boldsymbol \mu}\) equals \(\phi\) applied to the quadratic form \({\bf t}^T\Sigma {\bf t}\,,\) i.e., \(\varphi_{(\vec{\bX},\vec{\bY})-{\boldsymbol \mu}}({\bf t})=\phi({\bf t}^T\Sigma {\bf t})\) for all \({\bf t}\in \bR^{p+q}\,.\) 
For example, if \(\phi(u)=\exp(-u/2)\,,\) then \((\vec{\bX},\vec{\bY})\) is multivariate normally distributed with mean vector \({\boldsymbol \mu}\) and covariance matrix \(\Sigma\,.\) 
Elliptical distributions have a stochastic representation 
\begin{align*}
(\vec{\bX},\vec{\bY}) \eqd \boldsymbol{\mu} + R A \UU^{(k)}\,,
\end{align*}
where \(R\) is a non-negative random variable, \(A^TA=\Sigma\) is a full rank factorization of \(\Sigma\,,\) and where \(\UU^{(k)}\) is a uniformly on the unit sphere in \(\bR^k\) distributed random variable with \(k=\rank(\Sigma)\,;\) 
see \cite{Cambanis-1981} and \cite{Fang-1990} for more information on elliptical distributions. 
\\
Note that the dissimilarity functions of type \ref{tPD:D} and \ref{tMPD:D} are location-scale invariant (Corollary \ref{DF.Cor.InvBijT}), and thus, neither depend on the centrality parameter \(\boldsymbol{\mu}\) nor on componentwise scaling factors.

\begin{corollary}[Continuity for elliptical distributions]~~\label{Thm:Cont.Ell}\\
Consider the  \((p+q)\)-dimensional elliptically distributed random vectors \((\vec{\bX}_n,\vec{\bY}_n) \sim \mathcal{E}(\boldsymbol{\mu}_n,\Sigma_n,\phi_n)\), 
\(n\in \bN\), and \((\vec{\bX},\vec{\bY}) \sim \mathcal{E} (\boldsymbol{\mu},\Sigma,\phi)\), and assume \(\Sigma_{n}\), \(n\in \bN\), and \(\Sigma\) are positive definite. 
If \(\Sigma_n \to \Sigma\) and if either
\begin{enumerate}
\item  \(\phi_n =\phi\) for all \(n \in \mathbb{N}\) and if the radial part \(R\) associated with \(\phi\) has a continuous distribution function, or
\item \(\phi_n(u) \to \phi(u)\) for all \(u\geq 0\) and if the radial variable \(R_n\) associated with \(\phi_n\) has a density \(f_n\) such that
\begin{enumerate}
\item \((f_n)_{n\in \bN}\) is asymptotically uniformly equicontinuous on \((0,\infty)\,,\)
\item \((f_n)_{n\in \bN}\) is pointwise bounded, 
\end{enumerate}
\end{enumerate}
then $d^{p,q}(\vec{\bX}_{n} ,\vec{\bY}_{n})\xrightarrow{n\rightarrow \infty } d^{p,q}(\vec{\bX} ,\vec{\bY})$, where \(d^{p,q}\) is either a dissimilarity function of type \ref{tPD:D} or \ref{tMPD:D} according to \eqref{Diss.MOP:Represent.Cop}, \eqref{Diss.MOP:Represent.DCop} or \eqref{Diss.MOP:Represent.DCopAve}. 
\end{corollary}
\begin{proof}
The result is immediate from \cite[Proposition 4.4]{ansari_simple_2023}, the fact that elliptical distributions are closed under permutations, that every copula is continuous and the average of continuous function is also continuous.
\end{proof}

The following continuity result for the normal distribution is an immediate consequence of Corollary \ref{Thm:Cont.Ell}.
\begin{corollary}[Continuity for normal distributions]~\\
Consider the \((p+q)\)-dimensional normally distributed random vectors \((\vec{\bX}_n,\vec{\bY}_n) \sim \mathcal{N}(\boldsymbol{\mu}_n,\Sigma_n)\), 
\(n\in \bN\), and \((\vec{\bX},\vec{\bY}) \sim \mathcal{N}(\boldsymbol{\mu},\Sigma)\), and assume \(\Sigma_n\), \(n \in \bN\), and \(\Sigma\) are positive definite. 
Then \(\Sigma_n\to \Sigma\) implies $d^{p,q}(\vec{\bX}_{n} ,\vec{\bY}_{n})\xrightarrow{n\rightarrow \infty } d^{p,q}(\vec{\bX} ,\vec{\bY})$, where \(d^{p,q}\) is either a dissimilarity function of type \ref{tPD:D} or \ref{tMPD:D} according to \eqref{Diss.MOP:Represent.Cop}, \eqref{Diss.MOP:Represent.DCop} or \eqref{Diss.MOP:Represent.DCopAve}. 
\end{corollary}

Denote by \(t_\nu(\boldsymbol{\mu},\Sigma)\) the \(d\)-variate Student-t distribution with \(\nu>0\) degrees of freedom, symmetry vector \(\boldsymbol{\mu}\in \bR^d\) and symmetric, positive semi-definite \((d\times d)\)-matrix \(\Sigma\,.\) Then \(t_\nu(\boldsymbol{\mu},\Sigma)\) belongs to the elliptical class, where the radial variable has a density of the form \(g(t)= c[1+t/\nu]^{-(\nu+d)/2}\,,\) which is Lipschitz-continuous with Lipschitz constant \(\frac{\nu+d}{2\nu}\,.\) 
Hence, the following result is an immediate consequence of Corollary \ref{Thm:Cont.Ell}.

\begin{corollary}[Continuity for Student-t distributions]~\\
Consider the \((p+q)\)-dimensional t-distributed random vectors \((\vec{\bX}_n,\vec{\bY}_n) \sim t_{\nu_n}(\boldsymbol{\mu}_n,\Sigma_n)\), 
\(n\in \bN\), and \((\vec{\bX},\vec{\bY}) \sim t_\nu(\boldsymbol{\mu},\Sigma)\), 
and assume \(\Sigma_n\), \(n \in \bN\), and \(\Sigma\) are positive definite. 
Then \(\Sigma_n\to \Sigma\) and \(\nu_n\to \nu\) implies $d^{p,q}(\vec{\bX}_{n} ,\vec{\bY}_{n})\xrightarrow{n\rightarrow \infty } d^{p,q}(\vec{\bX} ,\vec{\bY})$, where \(d^{p,q}\) is either a dissimilarity function of type \ref{tPD:D} or \ref{tMPD:D} according to \eqref{Diss.MOP:Represent.Cop}, \eqref{Diss.MOP:Represent.DCop} or \eqref{Diss.MOP:Represent.DCopAve}. 
\end{corollary}

We further establish continuity within the class of \(l_1\)-norm symmetric distributions: \\
Denote by \({\bf S}_d\) a \(d\)-variate random vector that is uniformly distributed on the unit simplex \(\mathcal{S}_d = \{\xx\in \bR^d \mid \lVert \xx \rVert_1 =1\}\,.\)
A \(d\)-variate random vector \({\bf W}\) follows an \emph{\(l_1\)-norm symmetric distribution} if there exists a nonnegative random variable \(R\) independent of \({\bf S}_d\) such that \({\bf W} \eqd R\,{\bf S}_d\,.\) 
The following continuity result is immediate from \cite[Proposition 4.10]{ansari_simple_2023}:

\begin{corollary}[Continuity for \(l_1\)-norm symmetric distributions]~~\\
Consider the \((p+q)\)-dimensional \(l_1\)-norm symmetric random vectors \((\vec{\bX}_n,\vec{\bY}_n) \eqd R_n {\bf S}_{p+q}\), \(n\in \bN\), and \((\vec{\bX},\vec{\bY}) \eqd R \,{\bf S}_{p+q}\),
and assume \(F_{R_n}\) and \(F_R\) are continuous with \(F_{R_n}(0)=F_R(0)=0\) for all \(n \in \bN\).
Then \(R_n\xrightarrow{~d~} R\) implies $d^{p,q}(\vec{\bX}_{n} ,\vec{\bY}_{n})\xrightarrow{n\rightarrow \infty } d^{p,q}(\vec{\bX} ,\vec{\bY})$, where \(d^{p,q}\) is either a dissimilarity function of type \ref{tPD:D} or \ref{tMPD:D} according to \eqref{Diss.MOP:Represent.Cop}, \eqref{Diss.MOP:Represent.DCop} or \eqref{Diss.MOP:Represent.DCopAve}. 
\end{corollary}

From the above, we may conclude that the dissimilarity functions of type \ref{tPD:D} and \ref{tMPD:D} are continuous if certain conditions are satisfied.
We refer to \cite[Section 4]{ansari_simple_2023} for more results on the continuity of \(\kappa^{q|p}\) and concrete examples.

\section{Estimation}\label{Sec:Est}

In the present section estimators \(d^{p,q}_{\Pi,n}\) and \(d^{p,q}_{\mathrm{ave},n}\) for the type \ref{tPD:D} and \ref{tMPD:D} dissimilarity functions \(d^{p,q}_{\Pi}\) and \(d^{p,q}_{\mathrm{ave}}\) are introduced, both of which rely on the nearest neighbor based estimator \(\kappa^{q|p}_n\) for 
\(\kappa^{q|p}\) defined by \eqref{Def.MoP} and presented in \cite{ansari_simple_2023}.
The properties of \(\kappa^{q|p}_n\) imply strong consistency and a computation time of $O(n \log n)$ for the estimators \(d^{p,q}_{\Pi,n}\) and \(d^{p,q}_{\mathrm{ave},n}\), respectively.

\bigskip
Consider a $(p+q)$-dimensional random vector $(\vec{\bX},\vec{\bY})$ and i.i.d. copies $(\vec{\bX}_1,\vec{\bY}_1), \dots, (\vec{\bX}_n,\vec{\bY}_n)$. 
Recall that \((\vec{\bX},\vec{\bY})\) is assumed to have non-degenerate components, i.e., 
none of the random variables within \(\bX\) or \(\bY\) does follow a one-point distribution.
As estimators for \(d^{p,q}_{\Pi}\) and \(d^{p,q}_{\mathrm{ave}}\) we propose the statistics \(d^{p,q}_{\Pi,n}\) and \(d^{p,q}_{\mathrm{ave},n}\) given by
\begin{align} \label{Def.Est.}
  d^{p,q}_{\Pi,n}
  & := (1 - \kappa^{p|q}_n (\vec{\bX}|\vec{\bY})) (1 - \kappa^{q|p}_n(\vec{\bY}|\vec{\bX}))
  \\
  d^{p,q}_{\mathrm{ave},n}
  & := 1 - \frac{\kappa^{p|q}_n (\vec{\bX}|\vec{\bY}) + \kappa^{q|p}_n(\vec{\bY}|\vec{\bX})}{2}
 \end{align}
with $\kappa^{q|p}_n$ being the estimator proposed in \cite[Eq. (40)]{ansari_simple_2023} and given by the following plug-in construction principle
\begin{align} \label{Def.Est.2}
  \kappa^{q|p}_n(\vec{\bY}|\vec{\bX})
  & := \frac{1}{q!} \sum_{\sigma\in S_q} T^q_n(\sigma(\vec{\bY})|\vec{\bX})
  \\
  T^q_n(\vec{\bY}|\vec{\bX})
  &  = 1 - \frac{q - \sum_{i=1}^{q} T_n(Y_i | (\vec{\bX},Y_{i-1},\dots,Y_{1}))}{q - \sum_{i=1}^{q} T_n(Y_i | (Y_{i-1},\dots,Y_{1}))}\,,
  \notag
\end{align}
such that, for an i.i.d. sample $(\vec{\bZ}_1,Y_1), \dots, (\vec{\bZ}_n,Y_n)$ of a \((d+1)\)-dimensional random vector \((\vec{\bZ},Y)\), \(d \geq 1\),
\begin{align*}
  T_n(Y|\vec{\bZ})
  &  = \frac{\sum_{k=1}^{n} (n \, \min\{R_k,R_{N(k)}\}-L_k^2)}{\sum_{k=1}^{n} L_k \,(n - L_k)}\,,
\end{align*}
where $R_k$ denotes the rank of $Y_k$ among $Y_1, \dots, Y_n$, i.e., the number of $j$ such that $Y_j \leq Y_k$, 
and $L_k$ denotes the number of $j$ such that $Y_j \geq Y_k$ (cf. \cite{azadkia_simple_2021}).
For each $k$, the number $N(k)$ denotes the index $l$ such that $\vec{\bZ}_l$ is the nearest neighbour of $\vec{\bZ}_k$ with respect to the Euclidean metric on $\mathbb{R}^d$.
Since there may exist several nearest neighbours of $\vec{\bZ}_k$ ties are broken at random.

From the definitions of \(d^{p,q}_{\Pi,n}\) and \(d^{p,q}_{\mathrm{ave},n}\) in \eqref{Def.Est.}, 
it is immediately apparent that the estimators are built on the dimension reduction principle revealed in \cite{ansari_simple_2023} (see also \cite{chatterjee2021}) which is key to a fast estimation of \(d^{p,q}_{\Pi,n}\) and \(d^{p,q}_{\mathrm{ave},n}\).


In \cite{ansari_simple_2023}, the authors proved that \(\kappa^{q|p}_n\) in \eqref{Def.Est.2} is a strongly consistent estimator for \(\kappa^{q|p}\). As a direct consequence, we obtain strong consistency of \(d^{p,q}_{\Pi,n}\) and \(d^{p,q}_{\mathrm{ave},n}\).

\begin{theorem}[Consistency]\label{Est.Consistency}~\\
It holds that $\lim_{n \to \infty} d^{p,q}_{\Pi,n} = d^{p,q}_{\Pi}$ almost surely and $\lim_{n \to \infty} d^{p,q}_{\mathrm{ave},n} = d^{p,q}_{\mathrm{ave}}$ almost surely.
\end{theorem}

\begin{remark}\leavevmode
\begin{enumerate}
\item 
From the properties of the estimator \(\kappa^{q|p}_n\), see \cite[Remark 5.2]{ansari_simple_2023}, it follows that \(d^{p,q}_{\Pi,n}\) and \(d^{p,q}_{\mathrm{ave},n}\) can be computed in $O(n \log n)$ time.
 
\item 
The estimators \(d^{p,q}_{\Pi,n}\) and \(d^{p,q}_{\mathrm{ave},n}\) are model-free with no tuning parameters and are consistent under sparsity assumptions.
\end{enumerate}
\end{remark}

\section{Simulation study and real data examples}
\label{Sec:Data}

In an initial simulation study (Subsection \ref{Sim:DF.AB}) 
the different candidates for type \ref{tPD:D} and type \ref{tMPD:D} dissimilarity functions are evaluated, and it is found that choosing copulas as aggregating functions that are too close to the Fr{\'e}chet-Hoeffding upper bound \(M\) or too close to the Fr{\'e}chet-Hoeffding lower bound \(W\), respectively, results in poor performance.
Therefore, we suggest using \(d^{p,q}_{\Pi}\) as type \ref{tPD:D} and $d_{\mathrm{ave}}^{p,q}$ as type \ref{tMPD:D} dissimilarity function.
The agglomerative hierarchical variable clustering methods based on \(d^{p,q}_{\Pi}\) and $d_{\mathrm{ave}}^{p,q}$ are then compared to alternative hierarchical variable clustering methods available in the literature (Subsection \ref{Sim:Sub:Comp}), 
their resilience to noise is elaborated (Subsection \ref{SubSec:Data:Noise}),
and their performance is tested when the number of random variables to be clustered is large (Subsection \ref{Subsect:Sim:Large}).

\subsection{Evaluation of type \ref{tPD:D} and type \ref{tMPD:D} dissimilarity functions} \label{Sim:DF.AB}

We start with an empirical analysis and consider a data set of bioclimatic variables for $n=1,862$ locations homogeneously distributed over the global landmass from  CHELSA (\enquote{Climatologies at high resolution for the earth’s land surface areas}, \cite{karger_climatologies_2017}).
The variables in this data set can be split into two groups: 
thermal variables and precipitation based variables (for details see Table \ref{tab:clim_var}), making the data set particularly well suited for a cluster analysis.
\\
We are interested in whether the various candidates presented in Subsection \ref{SubSec:DF.AB} for choosing a dissimilarity function (both type \ref{tPD:D} and \ref{tMPD:D}) behave comparable or whether differences can be identified.
Therefore, three thermal variables \{AMT, MTCQ, MTDQ\} and three precipitation-based variables \{AP, PCQ, PDQ\} are selected and an agglomerative hierarchical clustering is performed, for which we expect a successful partitioning into the two groups of variables \{AMT, MTCQ, MTDQ\} and \{AP, PCQ, PDQ\}.
\\
For the evaluation of type \ref{tPD:D} dissimilarity functions, we employ $d_{C}^{p,q}$ where \(C\) belongs to either the Gaussian or the Gumbel copula family for varying parameter values.
\begin{figure}[p] 
\centering
    \subfigure[$d_C^{p,q}$ with Gaussian copula \(C\) and different choices of parameter (given in brackets).]{
        \begin{minipage}[b]{1\textwidth}
        \centering
        \includegraphics[width=1\textwidth]{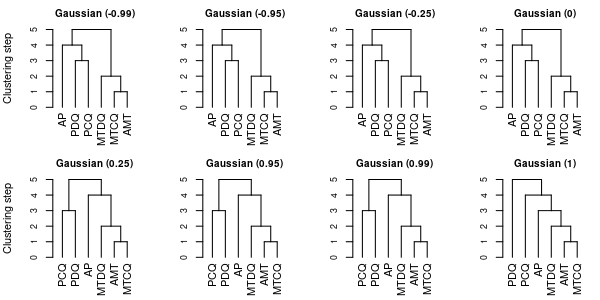} 
        \end{minipage}
    }
    \subfigure[$d_C^{p,q}$ with Gumbel copula \(C\) and different choices of parameter (given in brackets)]{
        \begin{minipage}[b]{1\textwidth}
        \centering
        \includegraphics{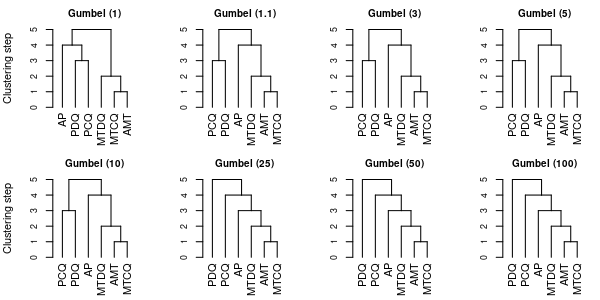} 
        \end{minipage}
    }
    \caption{Dendrograms of the hierarchical clustering in Subsection \ref{Sim:DF.AB} using type \ref{tPD:D} dissimilarity functions.}
    \label{fig:dissfunc.choice.A}
\end{figure}
From Figure \ref{fig:dissfunc.choice.A} we observe that, for either copula family, from a certain parameter threshold onwards a tendency to form chains occurs in the clustering output.
To confirm this observation, we conduct the investigation for three additional copula families (Clayton copula, Frank copula and Joe copula) and can conclude that, for each of these copula families, there is a tendency to form a chain when a higher parameter is chosen, i.e., when the aggregating copula is close to the Fr{\'e}chet-Hoeffding upper bound \(M\).
Table \ref{tab:threshold} presents, for this particular data set, thresholds for the parameter values and the Kendall's tau of the different copula families from which a tendency to form a chain emerges.
We observe that the Frank family makes the clustering result most prone to produce a chain, while the Gumbel family exhibits the least tendency for producing a chain.
\begin{table}[htbp]
  \centering
    \begin{tabular}{ccc}
    \toprule
    Copula family & Parameter & Kendall's tau \\
    \midrule
    Gaussian & 0.04  & 0.025 \\
    Gumbel   & 1.03  & 0.029 \\
    Clayton  & 0.04  & 0.020 \\
    Frank    & 0.13  & 0.014 \\
    Joe      & 1.05  & 0.028 \\
    \bottomrule
    \end{tabular}%
\caption{Threshold value for the parameter/Kendall's tau of various copula families, above which a tendency to form a chain in the dendrogram is visible for the data set analysed in Subsection \ref{Sim:DF.AB}.}
\label{tab:threshold}
\end{table}%
\\
For the evaluation of type \ref{tMPD:D} dissimilarity functions we employ $d_{C^*}^{p,q}$ with the Gaussian copula family for varying parameter values and $d^{p,q}_{\mathrm{ave}}$.
\begin{figure}[htbp] 
\centering
    \includegraphics{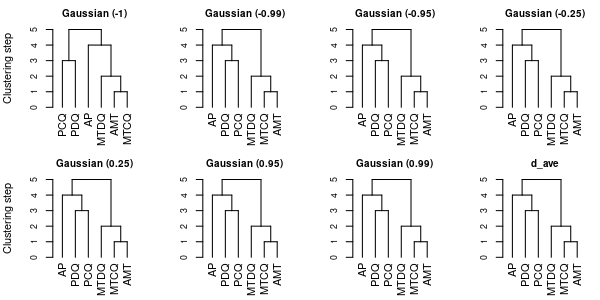} 
    \caption{Dendrograms of the hierarchical clustering in Subsection \ref{Sim:DF.AB} using type \ref{tMPD:D} dissimilarity functions $d_{C^*}^{p,q}$ (with Gaussian copula and different choices of parameter given in brackets) and $d_{\mathrm{ave}}^{p,q}$.}
    \label{fig:dissfunc.choice.B}
\end{figure} 
From Figure \ref{fig:dissfunc.choice.B} we observe that almost no difference in the clustering output for type \ref{tMPD:D} dissimilarity functions occurs.
In almost all the cases, i.e., when using $d_{C^*}^{p,q}$ with Gaussian copula \(C\) (parameter \(\geq -0.99\)) or $d_{\mathrm{ave}}^{p,q}$, the six variables are successfully partitioned into the two clusters
\{AMT, MTCQ, MTDQ\} and \{AP, PCQ, PDQ\}.
An important observation is the somewhat unsatisfactory clustering result when choosing the dissimilarity function $d_{C^*}^{p,q}$ for the Gaussian copula with parameter \(-1\), which here represents the Fr{\'e}chet-Hoeffding lower bound \(W\).

We now confirm and substantiate the observations just made about the behaviour of the various candidates for choosing a dissimilarity function through a simulation study.
Therefore, consider three sets of random variables \(\{N_1, N_2, N_3\}\), \(\{C_1, C_2, C_3\}\) and \(\{G_1, G_2, G_3\}\) distributed according to different copulas, namely Gaussian copula with Kendall's tau \(\tau = 0.15\), Clayton copula with Kendall's tau \(\tau = 0.3\) and Gumbel copula with Kendall's tau \(\tau = 0.45\). 
An agglomerative hierarchical clustering for all nine variables using different dissimilarity functions (both type \ref{tPD:D} and \ref{tMPD:D}) is performed, for which we expect a successful partitioning into the three groups. 
\begin{figure}[h] 
\centering
    \includegraphics{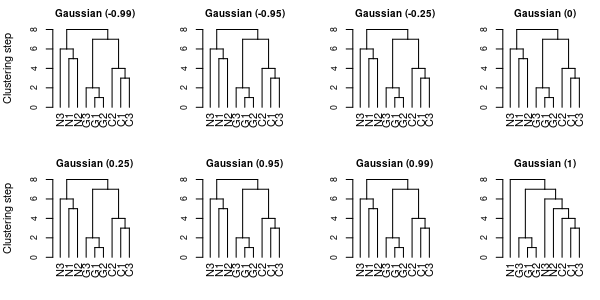}
    \caption{Dendrograms of the hierarchical clustering in Subsection \ref{Sim:DF.AB} using type \ref{tPD:D} functions $d_C^{p,q}$ with Gaussian copula for different choices of parameter (given in brackets).}
    \label{fig:dissfunc.choice.A.2}
\end{figure}
\begin{figure}[h] 
\centering
    \includegraphics{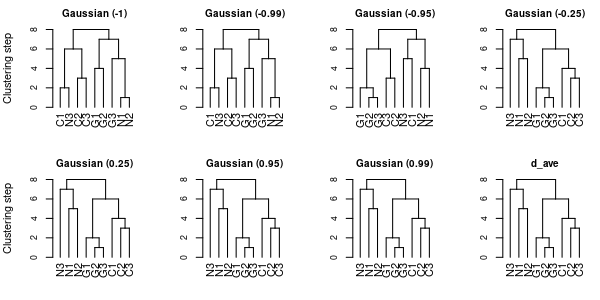}
    \caption{Dendrograms of the hierarchical clustering in Subsection \ref{Sim:DF.AB} using type \ref{tMPD:D} functions $d_{C^*}^{p,q}$ (with Gaussian copula and different choices of parameter given in brackets) and $d_{\mathrm{ave}}^{p,q}$.}
    \label{fig:dissfunc.choice.B.2}
\end{figure}
\\
From Figures \ref{fig:dissfunc.choice.A.2} and \ref{fig:dissfunc.choice.B.2} we observe that in most cases the nine variables can be successfully clustered.
Still, Figure \ref{fig:dissfunc.choice.A.2} depicts a tendency of type \ref{tPD:D} dissimilarity functions to form chains when the copula used for aggregating the two degrees of predictability is too close to (here: equals) the Fr{\'e}chet-Hoeffding upper bound \(M\).
Complementing what we observed before, the type \ref{tMPD:D} dissimilarity function \(d_{C^*}^{p,q}\), where \(C\) is a Gaussian copula, performs poorly whenever the parameter of the Gaussian copula is too small, i.e., whenever $C$ is too close to the Fr{\'e}chet-Hoeffding lower bound \(W\).
Unlike \(d_{C^*}^{p,q}\), \(d_{\mathrm{ave}}^{p,q}\) performs well on this data set: 
The nine variables are correctly clustered into three groups, where the three variables generated from the set \(\{G_1, G_2, G_3\}\) with highest Kendall's tau are clustered first, while the three variables generated from the set \(\{N_1, N_2, N_3\}\) with the lowest Kendall's tau are clustered last.

Building upon the above empirical analysis and simulation study, we conclude that for type \ref{tPD:D} dissimilarity functions it seems advisable to avoid copulas as aggregating functions that are too close to the Fr{\'e}chet-Hoeffding upper bound \(M\), otherwise the dendrogram is very likely to return a chain.
Therefore and due to its simple structure and appealing performance, we recommend (and in what follows also use) the independence copula \(\Pi\) as aggregating function and hence the dissimilarity function $d_\Pi^{p,q}$.
Instead, for type \ref{tMPD:D} dissimilarity functions, it seems advisable to avoid copulas as aggregating functions that are too close to the Fr{\'e}chet-Hoeffding lower bound \(W\), as they exhibit poor performance.
Therefore and due to its simple structure, appealing performance and the ability to characterize independence (Theorem \ref{MainThmDF}), we recommend (and in what follows also use) the dissimilarity function $d_{\mathrm{ave}}^{p,q}$.

\subsection{A simulation study comparing different hierarchical variable clustering methods}
\label{Sim:Sub:Comp}

We now discuss alternative hierarchical variable clustering methods available in the literature (based on different notions of multivariate similarity, cf. Subsection \ref{SubSec:Comparison}), 
including those based on \(\Phi\)-dependence \cite{gijbels2023} and
measures of multivariate concordance \cite{fuchs_dissimilarity_2021}, 
and elucidate their relation to hierarchical variable clustering based on (mutual) perfect dependence.
As type \ref{tPD:D} and type \ref{tMPD:D} dissimilarity functions 
we employ $d_{\Pi}^{p,q}$ and $d_{\mathrm{ave}}^{p,q}$, respectively, 
as measure of multivariate concordance we employ multivariate Spearman's footrule \cite{perez_nonparametric_2023, fuchs2019spearman},
while \(\Phi\)-dependence 
relies on mutual information \cite{gijbels2023,kojadinovic_agglomerative_2004} 
(measured by the Kullback–Leibler divergence \cite{kullback_information_1951}).

We employ the four above-mentioned hierarchical variable clustering methods for clustering five variables among which different types of association occur, including linear and nonlinear relationships as well as bivariate and multivariate dependencies.

Therefore, consider the set of random variables \(\bX = \{X_1,X_2,X_3,X_4,X_5\}\) with 
$X_1$ and $X_2$ being obtained from simulating a $2$-dimensional t-copula with correlation parameter $\rho=0$ and $\nu=0.1$ degrees of freedom, $X_3=-\mathrm{exp}(X_2)+\epsilon_1$ ($\epsilon_1 \sim\mathcal{N}(0, 0.2^2)$), $X_4=\log(-X_3)+\epsilon_2$ ($\epsilon_2\sim\mathcal{N}(0, 1)$), and $X_5=-\sin(1.5X_4+0.5X_2)$; 
Figures \ref{fig:Scatterplot} and \ref{fig:relation} depict and describe the (pairwise) dependence structures within \(\bX\), along with the expected clustering result if the focus lies on (mutual) perfect dependence.
\begin{figure}[htbp] 
    \centering
    \includegraphics[width=0.9\textwidth]{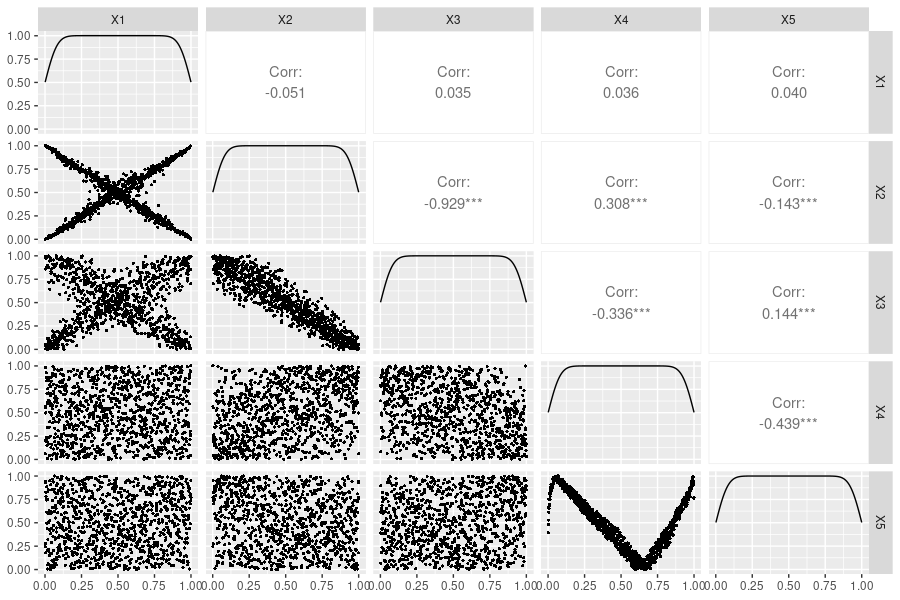}
    \caption{Pairwise dependence structures of the five random variables ($n = 1,000$) discussed in Subsection \ref{Sim:Sub:Comp}.}
    \label{fig:Scatterplot}
\end{figure}
\begin{figure}[htbp]
\centering
    \subfigure[Pairwise relationships between variables]{
        \begin{minipage}[b]{0.47\textwidth}
        \centering
        \includegraphics[width=1.11\textwidth]{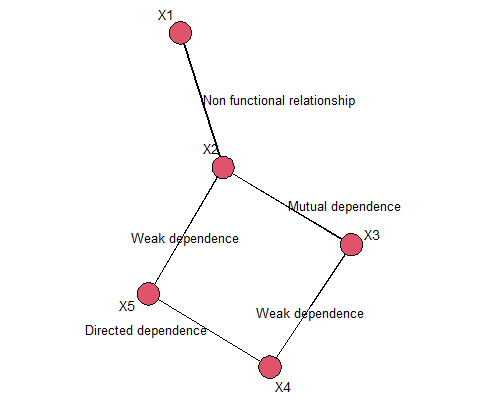} 
        \end{minipage}
    }
    \subfigure[Expected clustering result]{
        \begin{minipage}[b]{0.47\textwidth}
        \centering
        \includegraphics[width=1.11\textwidth]{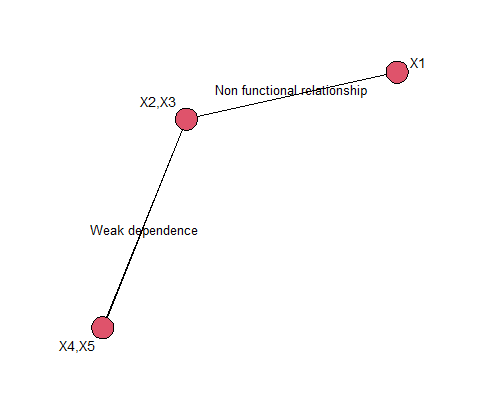} 
        \end{minipage}
    }
    \caption{Pairwise relationships and expected clustering result of the five variables ($n = 1,000$) discussed in Subsection \ref{Sim:Sub:Comp}.}
    \label{fig:relation}
\end{figure}
As depicted in Figure \ref{fig:relation} (a), there is a strong nonfunctional relationship between $X_1$ and $X_2$, a strong bi-directional relationship between $X_2$ and $X_3$, and a strong uni-directional relationship between $X_4$ and $X_5$. 
Weak dependence appears between $X_2$ and $X_5$ and between $X_3$ and $X_4$.

Before comparing the results obtained via the various hierarchical variable clustering methods, we first examine ways of determining an optimal number of clusters and hence an optimal partition.
At this point, recall that the focus of this paper is on clustering random variables and not data points, so that many optimality criteria known in literature are not applicable.
\\
Following \cite{kojadinovic_agglomerative_2004, gijbels2023, hansen_cluster_1997}, 
here an optimal partition is understood as one that maximizes the intra-cluster similarity and minimizes the inter-cluster similarity. 
The former criterion describes the homogeneity within the individual clusters and can be determined using the so-called \emph{average diameter} (Adiam) \cite{kojadinovic_agglomerative_2004, gijbels2023, hansen_cluster_1997} which, for a given partition into disjoint clusters \(\{\bX_1, \dots, \bX_l\}\), \(l\geq 2\), equals the arithmetic mean over the diameters \({\rm diam} (\bX_k)\) of the individual clusters \(\bX_k\), \(k \in \{1,\dots,l\}\), 
where 
\begin{align}\label{Def:ADiam}
    {\rm diam} (\bX)
    = \begin{cases}
      \underset{X,Y \in \bX}{\min} \quad \{1-d^{1,1}(X,Y)\}
      & \textrm{if } |\bX| > 1  
      \\
      1
      & \textrm{if } |\bX| = 1      
      \end{cases},
\end{align}
and \(|\bX|\) denotes the cardinality of \(\bX\). 
A larger Adiam indicates a higher homogeneity.
The latter criterion describes the separation of the different clusters and can be determined using the so-called \emph{maximum split} (Msplit) \cite{kojadinovic_agglomerative_2004, gijbels2023, hansen_cluster_1997}, the maximum over all \({\rm split} (\bX_k)\) of the individual clusters \(\bX_k\), \(k \in \{1,\dots,l\}\),
where 
\begin{align}\label{Def:MSplit}
    {\rm split} (\bX)
    = \underset{X \in \bX, Y \in \cX \backslash \bX}{\max} \quad \{1- d^{1,1}(X,Y)\}.
\end{align}
A smaller Msplit indicates a higher separation. 
For partitioning, we favour a high similarity within individual clusters and a low similarity among different clusters. 
Since during the hierarchical clustering procedure the homogeneity within the individual clusters decreases and the separation of the different clusters increases, an optimal partition comes at the best trade off of both quantities (see Figure \ref{fig:Partition} for an illustration).

\begin{remark}[Average diameter and maximum split]~~\label{Rem:ADiamMSplit} \\ 
Both criteria, average diameter and maximum split, use solely pairwise information of the random variables involved.
A more accurate but computationally demanding approach is to extend the definition of \({\rm diam} (\bX)\) to the maximum over all disjoint subsets \(\bX^\ast, \bY^\ast \subseteq \bX\) and that of \({\rm split} (\bX)\) to the minimum over all (disjoint) subsets \(\bX^\ast \subseteq \bX\) and \(\bY^\ast \subseteq \cX \backslash \bX\).
\end{remark}

\noindent
Another criterion for determining an optimal number of clusters is the so-called  \emph{Silhouette coefficient} (SC) going back to \cite{leonard_kaufman_finding_1990} and which is defined as a mixed coefficient involving the inter-cluster and intra-cluster dissimilarities.
An optimal number of clusters is then determined by computing the so-called silhoutte value for each object (here random variable), a value evaluating to what extend the considered object is clustered correctly, and averaging over all these values
(see Figure \ref{fig:Partition} for an illustration).

\begin{figure}[htbp]
\centering
    \subfigure[Type \ref{tPD:D}]{
        \begin{minipage}[b]{0.32\textwidth}
        \includegraphics[width=1\textwidth]{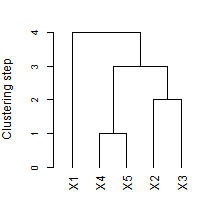} 
        \end{minipage}
        \begin{minipage}[b]{0.32\textwidth}
        \includegraphics[width=1\textwidth]{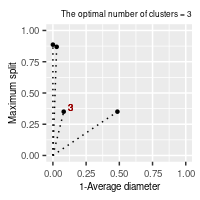}
        \end{minipage}
        \begin{minipage}[b]{0.32\textwidth}
        \includegraphics[width=1\textwidth]{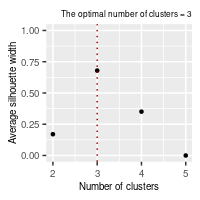}
        \end{minipage}
    }
    \subfigure[Type \ref{tMPD:D}]{
        \begin{minipage}[b]{0.32\textwidth}
        \includegraphics[width=1\textwidth]{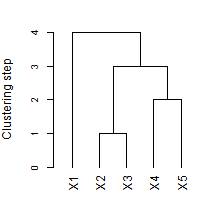} 
        \end{minipage}
        \begin{minipage}[b]{0.32\textwidth}
        \includegraphics[width=1\textwidth]{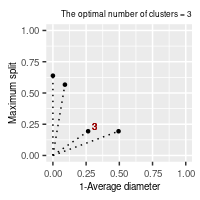}
        \end{minipage}
        \begin{minipage}[b]{0.32\textwidth}
        \includegraphics[width=1\textwidth]{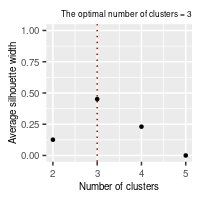}
        \end{minipage}
    }
    \subfigure[\(\Phi\)-dependence measure]{
        \begin{minipage}[b]{0.32\textwidth}
        \includegraphics[width=1\textwidth]{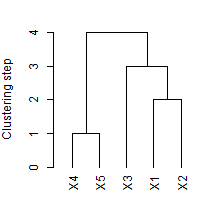} 
        \end{minipage}
        \begin{minipage}[b]{0.32\textwidth}
        \includegraphics[width=1\textwidth]{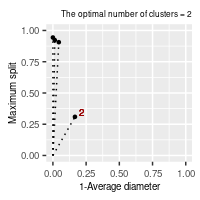}
        \end{minipage}
        \begin{minipage}[b]{0.32\textwidth}
        \includegraphics[width=1\textwidth]{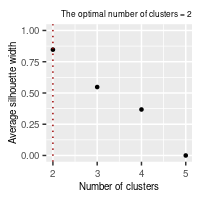}
        \end{minipage}
    }
    \subfigure[Multivariate Spearman's footrule]{
        \begin{minipage}[b]{0.32\textwidth}
        \includegraphics[width=1\textwidth]{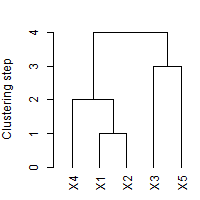} 
        \end{minipage}
        \begin{minipage}[b]{0.32\textwidth}
        \includegraphics[width=1\textwidth]{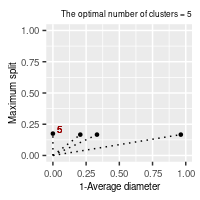}
        \end{minipage}
        \begin{minipage}[b]{0.32\textwidth}
        \includegraphics[width=1\textwidth]{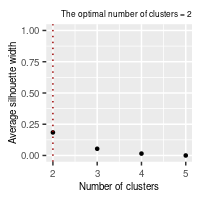}
        \end{minipage}
    }
    \caption{Dendrograms (left) and the corresponding trade-off between average diameter and maximum split (middle) \& Silhouette coefficient (right) for the simulation study performed in Subsection \ref{Sim:Sub:Comp} ($n=1,000$).}
    \label{fig:Partition}
\end{figure}

From Figure \ref{fig:Partition} we observe that the dendrograms obtained via type \ref{tPD:D} and \ref{tMPD:D} dissimilarity functions are quite similar with only one slight, but important and valid difference being that, in the case of type \ref{tPD:D}, variables $X_4$ and $X_5$ are clustered first whereas, in the case of type \ref{tMPD:D}, $X_2$ and $X_3$ are the first to cluster. 
No matter which of the two dissimilarity functions is used, $X_1$ is the last to be clustered.
The optimal partitions for type \ref{tPD:D} and type \ref{tMPD:D} are identical with the same optimal number of clusters, regardless which optimality criterion is used:
the five variables are partitioned into three clusters $\{ X_1\}$, $\{ X_2, X_3\}$, and $\{ X_4, X_5\}$, as expected (see Figure \ref{fig:relation}).
Instead, if $\Phi$-dependence is used, then the clustering result obtained is different from (but related to) the one above in such that $X_1$ is clustered together with $\{ X_2, X_3\}$. No matter which criterion is used, the optimal number of cluster is two resulting in the optimal partition $\{ X_1, X_2, X_3\}$ and $\{ X_4, X_5\}$.
In contrast to the three above observed and related optimal partitions, the clustering result obtained via multivariate Spearman's footrule exhibits no structural resemblance.
Even the optimal number of clusters is not unique: If Adiam and Msplit are used, the optimal partition is the finest partition, with each individual variable as a separate cluster. If instead the Silhouette coefficient is used then the optimal number of clusters is $2$. However, notice that the silhouette coefficient in this case is rather low meaning that the obtained partition may fail to be reasonable.

\begin{remark}
As seen in Figure \ref{fig:Partition}, the optimal number of clusters obtained via the silhouette coefficient tends to match the optimal number achieved with Adiam and Msplit, 
at least in those situations where the silhouette coefficient is large enough.
If, however, the silhouette coefficient is too small, it seems more advisable to rely on Adiam and Msplit for determining an optimal number of clusters.
This can be justified by the fact that in such a situation the silhouette coefficient is too close to \(0\) being the value reached for the finest partition.
\end{remark}

Since concordance measures such as multivariate Spearman's footrule are only capable of detecting monotone relationships among variables, the clustering result obtained for the given data is not very helpful.
In contrast, directed dependence concepts (such as type \ref{tPD:D} and type \ref{tMPD:D}) and $\Phi$-dependence are capable of detecting non-monotone relationships, resulting in convincing clustering results for the given data set.
However, it is important to recognise that the concepts considerably differ in terms of the extent to which functional and non-functional relationships are to be identified (cf. Subsection \ref{SubSec:Comparison}), and thus also in terms of the clustering result achieved for the given data set.

\subsection{A simulation study about noise resistance}
\label{SubSec:Data:Noise}

Building upon the continuity of type \ref{tPD:D} and \ref{tMPD:D} dissimilarity functions as shown in Section \ref{Sec:Prop.}, we now illustrate the resilience of the corresponding agglomerative hierarchical clustering algorithms to noise.

Therefore, consider the set of random variables \(\bX = \{X_1,X_2,X_3,X_4,X_5,X_6\}\) with 
$X_1 \sim \mathcal{N}(0,1)$
$X_2 = X_1^2 + X_1 + \varepsilon_2$,
$X_3 \sim \mathcal{N}(0,1)$,
$X_4 = \exp(-X_3) + \varepsilon_4$, 
$X_5 = X_4 + \sin(X_3) + \varepsilon_5$ and
$X_6 \sim \mathcal{N}(0,1)$,
where $\varepsilon_2, \varepsilon_4, \varepsilon_5 \sim \mathcal{N}(0,\sigma^2)$ with varying variance $\sigma \geq 0$. 
When focusing on (mutual) perfect dependence, the benchmark classification is given by: 
Cluster $1$ = $\{ X_1, X_2\}$, Cluster $2$ = $\{ X_3, X_4, X_5\}$, and Cluster $3$ = $\{X_6\}$.
For the case when no noise is present in the data, i.e., \(\sigma=0\), 
Figure \ref{fig:without_noise} depicts the clustering results obtained via type \ref{tPD:D} and \ref{tMPD:D} dissimilarity functions. 
There, it can be seen that the five variables are successfully clustered and the optimal partition coincides with the benchmark partition.

\begin{figure}[h] 
\centering
    \subfigure[Type \ref{tPD:D}]{
        \begin{minipage}[b]{0.32\textwidth}
        \includegraphics[width=1\textwidth]{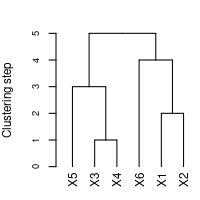} 
        \end{minipage}
        \begin{minipage}[b]{0.32\textwidth}
        \includegraphics[width=1\textwidth]{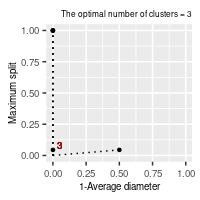}
        \end{minipage}
        \begin{minipage}[b]{0.32\textwidth}
        \includegraphics[width=1\textwidth]{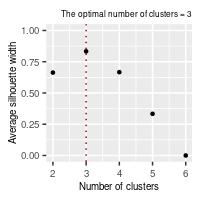}
        \end{minipage}
    }
    \subfigure[Type \ref{tMPD:D}]{
        \begin{minipage}[b]{0.32\textwidth}
        \includegraphics[width=1\textwidth]{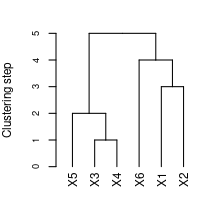} 
        \end{minipage}
        \begin{minipage}[b]{0.32\textwidth}
        \includegraphics[width=1\textwidth]{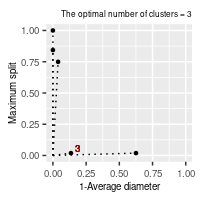}
        \end{minipage}
        \begin{minipage}[b]{0.32\textwidth}
        \includegraphics[width=1\textwidth]{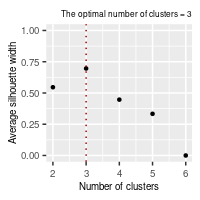}
        \end{minipage}
    }
    \caption{Dendrograms (left) and the corresponding trade-off between average diameter and maximum split (middle) \& Silhouette coefficient (right) for the simulation study performed in Subsection \ref{SubSec:Data:Noise} ($n=1,000$).}   
    \label{fig:without_noise}
\end{figure}

In what follows, we study the performance of the agglomerative hierarchical clustering algorithms when more and more noise is added to the data, i.e., when \(\sigma\) increases.
Thereby, we first analyse the resilience to noise of the clustering result and, in a second step, the resilience to noise of the optimality criteria.
Before that, however, let us briefly review tools that allow a comparison of the optimal partition obtained via the hierarchical clustering algorithms with the benchmark partition.
For determining the degree of congruence between two partitions, 
two popular indices exist: the Rand index (RI) \cite{rand_objective_1971} and the Fowlkes–Mallows index (FMI) \cite{fowlkes_method_1983}.
Both, Rand index and Fowlkes-Mallows index, rely on the number of same or different variables in the two partitions $A_1, A_2$. 
More precisely, when defining 
\begin{itemize}
    \item[-] TP as the number of pairs of objects that are in the same cluster in $A_1$ and in the same cluster in $A_2$,
    \item[-] FP as the number of pairs of objects that are in the same cluster in $A_1$ and in different clusters in $A_2$,
    \item[-] FN as the number of pairs of objects that are in different clusters in $A_1$ and in the same cluster in $A_2$,
    \item[-] TN as the number of pairs of objects that are in different clusters in $A_1$ and in different clusters in $A_2$,
\end{itemize}
Rand index (RI) and Fowlkes-Mallows index (FMI) are given as follows:
\begin{itemize}
    \item $RI = \frac{TP + TN}{TP + FP + FN + TN} $,
    \item $FMI = \sqrt{\frac{TP}{TP+FP}\frac{TP}{TP+FN}}$.
\end{itemize}
RI and FMI range from $0$ to $1$. A higher value for RI or FMI indicates a larger degree of congruence between the two partitions. 

\bigskip
\noindent
\emph{Resilience to noise regarding the clustering result:} \\
In a first simulation study, we investigate how much noise can be added to the data without changing the resulting partition.  
By following the benchmark partition we thereby set the optimal number of clusters to \(3\).
For several choices of \(\sigma\), we then compare the resulting partition with the benchmark partition and calculate their degree of congruence by means of the Rand index (RI) and the Fowlkes-Mallows index (FMI) and repeat each szenario $B = 100$ times.
Figure \ref{fig:error} depicts boxplots of the obtained values for RI and FMI.
From Figure \ref{fig:error} we observe that up to \(\sigma=4\) the clustering partition obtained via type \ref{tPD:D} and type \ref{tMPD:D} dissimilarity functions perfectly match the benchmark partition, revealing a high resilience to noise for both procedures.

\begin{figure}[htbp] 
    \centering
    \includegraphics[width=0.9\textwidth]{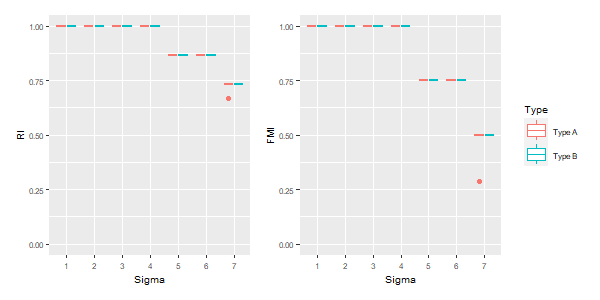}
    \caption{Boxplots of RI and FMI values (y-axis) for varying \(\sigma \in \{1,2,3,4,5,6,7 \} \) (x-axis) when investigating the resilience to noise with regard to the clustering result (Subsection \ref{SubSec:Data:Noise}).}
    \label{fig:error}
\end{figure}

\bigskip
\noindent
\emph{Resilience to noise regarding the optimal number of clusters:} \\
In a second simulation study we examine how much noise can be added to the data without changing the optimal number of clusters found for the benchmark partition.
Therefore, for several choices of \(\sigma\), we compute the optimal number of clusters obtained
(1) via the best trade off between average diameter and maximum split defined by \eqref{Def:ADiam} and \eqref{Def:MSplit} (bivariate/pairwise version),
(2) via the best trade off between average diameter and maximum split as described in Remark \ref{Rem:ADiamMSplit} (multivariate version) and
(3) via the Silhouette coefficient, and repeat the procedure $B = 100$ times.
Figure \ref{fig:num_AM&SC} depicts boxplots of the obtained optimal number of clusters.
From Figure \ref{fig:num_AM&SC} we observe that, compared to the clustering result, the optimality criteria are less resistant to noise no matter which dissimilarity function is used (type \ref{tPD:D} or type \ref{tMPD:D}).
Interestingly, the Silhouette coefficient performs more stable (up to \(\sigma=2.3\)) than average diameter and maximum split (pairwise, up to \(\sigma=0.3\)) and average diameter and maximum split (multivariate, up to \(\sigma=0.4\)), while, surprisingly, the latter two optimality criteria do not differ much, which is in favour of the pairwise version in terms of computation time.

\begin{figure}[htbp] 
    \centering
    \includegraphics[width=0.9\textwidth]{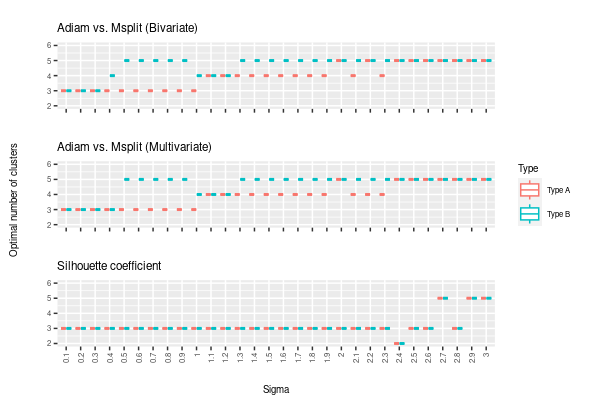}
    \caption{Boxplot of the optimal number of clusters (y-axis) using Adiam vs. Msplit and Silhouette coefficient for varying \(\sigma \in [0,3] \) (x-axis) when investigating the resilience to noise with regard to the optimal number of cluster (Subsection \ref{SubSec:Data:Noise}).}
    \label{fig:num_AM&SC}
\end{figure}

\subsection{A simulation study about clustering a larger number of variables}
\label{Subsect:Sim:Large}

Finally, we illustrate the performance of the agglomerative hierarchical clustering algorithm based on type \ref{tPD:D} and type \ref{tMPD:D} dissimilarity functions when the number of random variables to be clustered is large.

Therefore, consider the set of random variables \(\bX = \bX_A \cup \bX_B \cup \bX_C \cup \bX_D\) consisting of four independent sets of variables, namely 
\(\bX_A = \{A_1, \dots, A_5\}\), \(\bX_B = \{B_1, \dots, B_5\}\), \(\bX_C = \{C_1, \dots, C_5\}\), \(\bX_D = \{D_1, \dots, D_5\}\).
Given $\alpha \in \{0.4, 0.6, 0.8, 1\}$, 
each of the four sets is simulated from a $5$-dimensional Clayton copula with Kendall's tau value
$\tau_A = \alpha \cdot 0.2$, $\tau_B = \alpha \cdot 0.4$, $\tau_C = \alpha \cdot 0.6$, and $\tau_D = \alpha \cdot 0.8$, respectively. 
We then compare the resulting partition with the benchmark partition and calculate their degree of congruence by means of the Rand index (RI) and the Fowlkes-Mallows index (FMI) and repeat each szenario \(B = 100\) times. 
Figure \ref{fig:Larger} depicts boxplots of the obtained values for RI and FMI, and Figure \ref{fig:20var} shows the results for a single run.
We observe that, regardless of whether type \ref{tPD:D} or \ref{tMPD:D} dissimilarity function is used, the clustering result becomes more and more accurate (compared to the benchmark partition) as $\alpha$ increases. 
Clearly, the stronger the correlation between the variables (greater Kendall tau value of the copula) the earlier in the clustering process the variables are grouped into a cluster.

\begin{figure}[htbp] 
    \centering
    \includegraphics[width=0.9\textwidth]{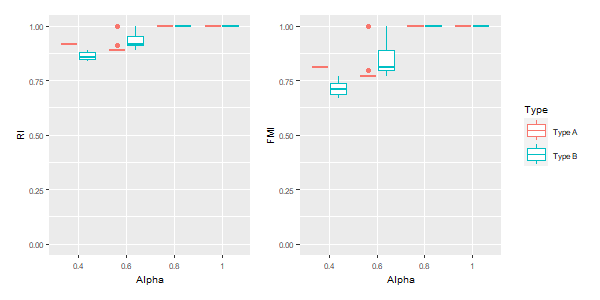}
    \caption{Boxplots of RI and FMI values (y-axis) for varying \(\alpha \in \{0.4, 0.6,0.8,1 \} \) (x-axis) for the simulation study performed in Subsection \ref{Subsect:Sim:Large}}
    \label{fig:Larger}
\end{figure}

\bigskip
\noindent {\bf Acknowledgement:} \\
Both authors gratefully acknowledge the support of the Austrian Science Fund (FWF) project
{P 36155-N} \emph{ReDim: Quantifying Dependence via Dimension Reduction}
and the support of the WISS 2025 project 'IDA-lab Salzburg' (20204-WISS/225/197-2019 and 20102-F1901166-KZP).  

\bigskip


\clearpage
\section{Supplementary material: Tables and figures} 
\label{appx:t_f}

\vspace*{\fill}
\begin{table}[!htbp]
\centering

    \begin{tabular}{@{}cccc@{}}
    \toprule
    Index & Variables & Description                                              & Category      \\ 
    \midrule
    1  & AMT       & Annual Mean Temperature {[}°C*10{]}                      & Temperature   \\
    2  & MTWeQ     & Mean Temperature of Wettest Quarter {[}°C*10{]}          & Temperature   \\
    3  & MTDQ      & Mean Temperature of Driest Quarter {[}°C*10{]}           & Temperature   \\
    4  & MTWaQ     & Mean Temperature of Warmest Quarter {[}°C*10{]}          & Temperature   \\
    5  & MTCQ      & Mean Temperature of Coldest Quarter {[}°C*10{]}          & Temperature   \\
    6  & MTWM      & Max Temperature of Warmest Month {[}°C*10{]}             & Temperature   \\
    7  & MTCM      & Min Temperature of Coldest Month {[}°C*10{]}             & Temperature   \\
    8  & TAR       & Temperature Annual Range {[}°C*10{]}                     & Temperature   \\
    9  & MDR       & Mean Diurnal Range {[}°C{]}                              & Temperature   \\
    10 & IT        & Isothermality (MDR/TAR)                                  & Temperature   \\
    11 & TS        & Temperature Seasonality {[}standard deviation{]}         & Temperature   \\
    12 & AP        & Annual Precipitation {[}mm/year{]}                       & Precipitation \\
    13 & PWeQ      & Precipitation of Wettest Quarter {[}mm/quarter{]}        & Precipitation \\
    14 & PDQ       & Precipitation of Driest Quarter {[}mm/quarter{]}         & Precipitation \\
    15 & PWaQ      & Precipitation of Warmest Quarter {[}mm/quarter{]}        & Precipitation \\
    16 & PCQ       & Precipitation of Coldest Quarter {[}mm/quarter{]}        & Precipitation \\
    17 & PWM       & Precipitation of Wettest Month {[}mm/month{]}            & Precipitation \\
    18 & PDM       & Precipitation of Driest Month {[}mm/month{]}             & Precipitation \\
    19 & PS        & Precipitation Seasonality {[}coefficient of variation{]} & Precipitation \\
    \bottomrule
    \end{tabular}
\caption{Bioclimatic variables in global climate data set from CHELSA used in Subsection \ref{Sim:DF.AB}}
\label{tab:clim_var}
\end{table}
\vspace*{\fill}

\begin{figure}[htbp] 
    \centering
    \subfigure[Type \ref{tPD:D}]{
        \begin{minipage}[b]{0.33\textwidth}
            \includegraphics[width=1\textwidth]{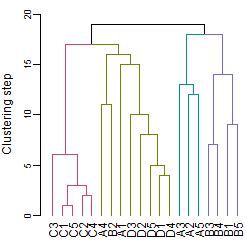} \\
            \includegraphics[width=1\textwidth]{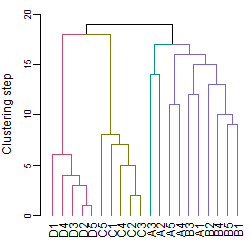} \\
            \includegraphics[width=1\textwidth]{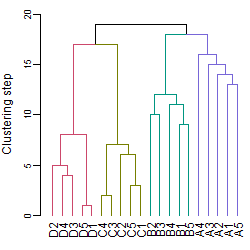} \\
            \includegraphics[width=1\textwidth]{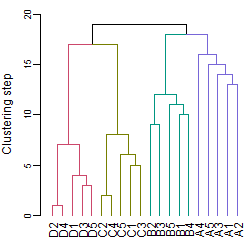} 
        \end{minipage}
    }
    \subfigure[Type \ref{tMPD:D}]{
        \begin{minipage}[b]{0.33\textwidth}
            \includegraphics[width=0.944\textwidth]{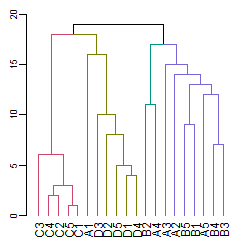} \\
            \includegraphics[width=0.944\textwidth]{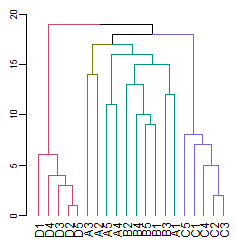} \\
            \includegraphics[width=0.944\textwidth]{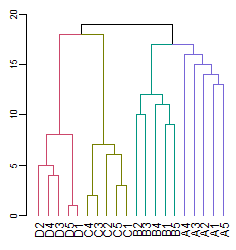} \\
            \includegraphics[width=0.944\textwidth]{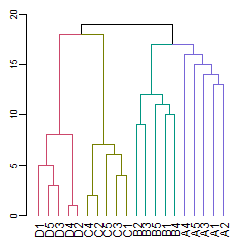}
        \end{minipage}
    }
\caption{Dendrograms of the hierarchical clustering of the 20 random variables in Subsection \ref{Subsect:Sim:Large} with varying \(\alpha \in \{0.4, 0.6,0.8,1\} \) (panels by row) ($n=500$).}
\label{fig:20var}
\end{figure}

\end{document}